\documentclass[a4paper,11pt]{article}

\usepackage[utf8]{inputenc}
\usepackage[T1]{fontenc}
\usepackage[english]{babel}
\usepackage[utopia,cal]{mathdesign}\usepackage{amsmath,amsthm} 
\usepackage{erewhon}
\usepackage[usenames]{xcolor}
\usepackage{graphicx}
\usepackage{tikz}

\usepackage{natbib}
\usepackage{setspace}
\renewcommand{\theenumi}{\roman{enumi}}
\usepackage{enumitem}\setlist[enumerate]{label={\rm(\theenumi)}}

\theoremstyle{plain}
\newtheorem{theorem}{Theorem}
\newtheorem{proposition}{Proposition}

\newtheorem{lemma}{Lemma}

\newtheorem{claim}{Claim}

\newtheorem*{theoremLR}{Quasiconcavification Theorem}
\newtheorem*{lemmaLR}{Lemma LR}

\theoremstyle{definition}
\newtheorem{definition}{Definition}

\DeclareMathOperator*{\argmax}{arg\,max}
\DeclareMathOperator*{\co}{co}
\newcommand{\supp}[1]{\operatorname{supp}{#1}}
\DeclareMathOperator{\inte}{int}

\newcommand{\abs}[1]{\lvert#1\rvert}

\newcommand{\norm}[1]{\lVert#1\rVert}
\newcommand{\set}[2]{\{\,{#1}:{#2}\,\}}
\newcommand{\setd}[2]{\left\{\,{#1}:{#2}\,\right\}}

\newcommand{\R}{\mathbb{R}}
\newcommand{\one}{\mathbf{1}}
\newcommand{\vsum}{\one\cdot}

\newcommand{\cD}{{\mathcal D}}
\newcommand{\cR}{{\mathcal R}}
\newcommand{\cU}{{\mathcal U}}
\newcommand{\cV}{{\mathcal V}}

\usepackage{hyperref} 

\title{Robust equilibria in cheap-talk games\\ with fairly transparent motives\thanks{We thank Frederic Koessler, Ming Li, Christina Pawlowitsch, Denis Shishkin, Joel Sobel, Elias Tsakas, and Andy Zapechelnyuk for helpful discussions and suggestions. Jan-Henrik Steg gratefully acknowledges financial support from the Austrian Science Fund (FWF) [Lise Meitner grant M~2905-G] and from the German Research Foundation (DFG) [SFB 1283/2 2021 – 317210226]. Christoph Kuzmics gratefully acknowledges financial support from the Jubil{\"a}umsfonds of the {\"O}sterreichische Nationalbank.}}

\author{Jan-Henrik Steg\thanks{Bielefeld University, Germany, jsteg@uni-bielefeld.de} \and Elshan Garashli\thanks{University of Graz, Austria, elshan.garashli@uni-graz.at} \and Michael Greinecker\thanks{ENS Paris-Saclay, France, michael.greinecker@ens-paris-saclay.fr} \and Christoph Kuzmics\thanks{University of Graz, Austria, christoph.kuzmics@uni-graz.at}}

\date{March 15, 2024}

\begin{document}

\maketitle

\begin{abstract}
For cheap-talk games with a binary state space in which the sender has state-independent preferences, we characterize equilibria that are robust to introducing slight state-dependence on the side of the sender. Not all equilibria are robust, but the sender-optimum is always achieved at some robust equilibrium.
\end{abstract}



\section{Introduction}

A sender learns about a state of the world and can provide some information about the state to a receiver who then chooses an action. Both the sender and the receiver care about which action is chosen in which state. The sender's problem has a particularly simple solution if the sender can commit to an information structure before learning the state. For that case, \citet{kamenica2011bayesian} have shown that the highest payoff the sender can achieve can be calculated as the concave envelope of a suitable value function for the sender, an approach pioneered by \citet{AumannMaschler1995}. When the sender cannot commit, we are in the more complex cheap-talk setting introduced by \citet{crawford1982strategic}. Finding sender-optimal equilibria in general cheap-talk games is an intractable problem. However, \citet{lipnowski2020cheap} have shown that the problem has a tractable solution when the sender's preferences over actions do not depend on the state, i.e., when the sender's motives are fully transparent. Then the highest payoff the sender can achieve is given by the quasiconcave envelope of the sender's value function instead of the concave envelope. This characterization makes it possible to calculate the value of commitment; it is the difference between the two envelopes.

\citet{lipnowski2020cheap} have even obtained a characterization of the whole set of sender equilibrium payoffs when the sender has transparent motives, not only the highest one. This set is an interval, so there is either a unique equilibrium payoff for the sender (from the trivial no-communication---or ``babbling''---equilibrium) or a continuum.

We are interested in the question which of the equilibria are sensitive to the main structural assumption in \citet{lipnowski2020cheap}. Which equilibria disappear as soon as arbitrarily little state dependence is introduced? And will the useful geometric characterization from \citet{lipnowski2020cheap} then still hold? To wit, \citet{diehl2021non} have shown that no nontrivial equilibrium in the class of transparent cheap-talk games studied by \citet{chakraborty2010persuasion} is robust against introducing slight uncertainty that the receiver has about the sender's---still state-independent---preferences; no equilibrium is robust in the sense of \citet{harsanyi1967games}.

We consider a different form of robustness. The information structure remains the same as in the reference model; the receiver knows the sender's preferences. Now it is we, the analysts, who are unsure about the sender's preferences; and now we only know they are close to state-independent. Indeed, we often know only the sender's ordinal ranking over the receiver's actions with certainty. We do not allow for perturbations that generate uncertainty of the receiver over the sender's preferences as in \cite{diehl2021non}, or perturbations that endow the receiver with private information as in \citet*{arieli2023robust}. Under the latter perturbations, but keeping preferences fixed, \citet*{arieli2023robust} show that sender-optimal equilibria are robust exactly when they are uninformative or completely reveal a state with positive probability. As we show, under our perturbations the optimal sender payoff is always robustly attainable.

For general games, the question of robustness with respect to small payoff perturbations was first asked by \citet{WuJiang1962}, who introduced the concept of an essential equilibrium. We adapt the notion of an essential equilibrium to cheap-talk games by restricting the allowable perturbations to those that keep cheap talk cheap.

Not all equilibria in transparent cheap-talk games with two states are robust.\footnote
{In contrast, \citet[Section 3]{farrell1993meaning} shows why refinements based on the \citet{kohlberg1986strategic} notion of strategic stability have almost no bite in cheap talk games. This includes refinements for signaling games such as the intuitive criterion of \citet{cho1987signaling}.
\citet{farrell1993meaning} also introduced a refinement custom-made for cheap-talk games, \emph{neologism-proof equilibrium}, which requires that no subset of sender types is self-signaling. A subset of sender types is self-signaling if all types within the subset and no types outside this subset would benefit from the receiver learning that the sender type is within this subset. Neologism-proofness has almost no bite in transparent cheap-talk games: Since all types of the sender have the same preferences, the only possible self-signaling set comprises all types, and signaling that the type is actually in this set benefits the sender only if the payoff from the equilibrium in question is less than the babbling equilibrium payoff.

A different approach to refining equilibria is introduced by \citet*{balkenborg2013refined} that has the potential to remove some mixed equilibria in cheap-talk games as shown in \citet*[Section 6]{balkenborg2015refined} for a very simple class of cheap-talk games. We don't know what this approach would deliver for the transparent cheap-talk games studied here. \citet*{lipnowski2022perfect} use a similar approach to refine equilibria in the Bayesian persuasion setting (i.e., with commitment) of \citet{kamenica2011bayesian}.\label{fn:refine}} However, the sender optimal payoff is always achieved at some robust equilibrium; the perturbations do not hurt the sender. In particular, the quasiconcave envelope of the sender's value function, as introduced by \citet{lipnowski2020cheap}, gives the highest sender payoff that can be robustly obtained. 

Our precise robustness notion is somewhat subtle. Allowing for all sufficiently small perturbations would destroy every nontrivial equilibrium; only the no-communication payoff would persist. Yet, we find that certain nontrivial equilibria are robust to open sets of perturbations that can be arbitrarily close to zero. Those equilibria are not ``knife-edge'' ones, they exist not only for a thin set of games. For the case of a binary state space, we are able to fully characterize the equilibria that are robust in this sense. We then use this characterization to construct robust equilibria that yield the sender-optimal payoff.\bigskip

The paper proceeds as follows: Section~\ref{sec:model} provides the model of finite cheap-talk games, our notion of equilibrium and robustness, and the benchmark results of \citet{lipnowski2020cheap}. Section~\ref{sec:examples} illustrates the logic of the main result of this paper in terms of two simple examples. Section~\ref{sec:results} provides the central characterization, Proposition~\ref{prop:binarybeliefrobust}, and the statement of the main result of the paper, Theorem~\ref{thm:smax}. The actual link between Proposition~\ref{prop:binarybeliefrobust} and Theorem~\ref{thm:smax} is formed by a second, more graphical characterization derived in Section~\ref{sec:graphical} as Proposition~\ref{prop:srobust}. In Section~\ref{sec:discussion}, we further discuss our approach. Two appendices contain several auxiliary results (Appendix~\ref{app:aux}) and all nontrivial proofs of results from the main text (Appendix~\ref{app:proofs}).

\section{The model}\label{sec:model}

We restrict ourselves to finite cheap-talk games. There are two players, a sender and a receiver. There is a finite state space $\Theta$ and a prior distribution $\mu_0\in \Delta\Theta$  with full-support.\footnote
{For any finite (or metrizable) space $X$, the set of all (Borel) probability measures over $X$ is denoted by $\Delta X$.}
Even though our main results assume that the state space is binary, we do not yet need this assumption in this section.

The sender first observes the true state and then sends a message from a finite set of messages. It will turn out that the exact message space is irrelevant, provided there are enough messages. The receiver observes the sender's message and takes an action from the finite set $A$.

The sender's and receiver's payoffs are given by (von Neumann--Morgenstern) utility functions $u_S:A\times\Theta  \to \R$ and $u_R:A\times\Theta \to \R$, respectively. Messages are not payoff-relevant; talk is cheap. A \emph{cheap-talk game} is then given by the quintuple $\Gamma=\langle\Theta,\mu_0,A,u_S,u_R\rangle$. To rule out trivial cases, we assume there are at least two states and at least two actions.

Throughout, we adopt the belief-based approach, which concentrates on the \emph{ex ante} distribution of the receiver's posterior beliefs that is induced by any combination of a sender strategy (for sending messages) with a belief system (specifying which posterior belief the receiver adopts after observing any particular message).\footnote
{The belief-based approach was pioneered by \citet{AumannMaschler1995} and had a recent renaissance thanks to the work of \citet{kamenica2011bayesian}. It is also the basis that \citet{lipnowski2020cheap} build on.}
Instead of specifying message choices, we let the sender directly choose a posterior belief distribution $p \in \Delta \Delta \Theta$ with finite support. That such a posterior belief distribution $p$ can be induced by some strategy to send messages and a corresponding belief update following Bayes' rule is equivalent to $p$ being \emph{Bayes plausible} \citep[in the language of][]{kamenica2011bayesian}, which means that the expected posterior equals the prior:
\begin{equation*}
\sum_{\mu \in \supp{p}} p(\mu) \mu = \mu_0.
\end{equation*}
For any $\mu \in \Delta \Theta$, let $A(\mu)$ be the receiver's optimal actions given belief $\mu$, that is
\begin{equation*}
A(\mu)=\argmax_{a\in A}\sum_{\theta \in \Theta} u_R(a,\theta)\mu(\theta).
\end{equation*}
Consider a receiver strategy $\rho\colon\supp{p}\to\Delta A$ that assigns a randomized action to any posterior belief that is possible given $p$. That $\rho$ is optimal for the receiver is equivalent to $\supp{(\rho(\mu))}\subseteq A(\mu)$ for all $\mu\in\supp{p}$.

We are interested in the sender's \emph{interim} expected payoff, the conditional expected payoff after learning the state but before releasing any information to the receiver. Often, we simply call it the payoff.\footnote
{In the Bayesian persuasion setting of \citet{kamenica2011bayesian}, where the sender and receiver always share the same posterior belief, and the cheap-talk setting with (completely) transparent motives of \citet{lipnowski2020cheap}, where the sender's expected payoff does not change when learning the state, it is enough to focus on the \emph{ex ante} payoff---a single number.}
Cheap talk is characterized by the inability of the sender to commit to any postulated strategy. Therefore, after learning the state, the sender can induce the receiver to adopt any of the ex ante possible posterior beliefs, but an equilibrium requires that the sender cannot benefit from manipulating any posterior belief generated by $p$. For a convenient expression of the sender's incentives in any state $\theta$, we enumerate the available actions arbitrarily, $A=\{a_1,\dots,a_n\}$. We then let (with a slight abuse of notation)
\begin{equation*}
u_S(\theta):=\left(u_S(a_1,\theta),\dots,u_S(a_n,\theta)\right)^\top
\end{equation*}
and similarly treat any mixed action $r\in\Delta A$ as the vector $(r(a_1),\dots,r(a_n))^\top$. We can then simply write $u_S(\theta)\cdot\rho(\mu)$ for the sender's interim expected payoff from any posterior belief $\mu$ and strategy $\rho$ for the receiver.

\begin{definition}\label{def:eqm}
For a cheap-talk game $\Gamma=\langle\Theta,\mu_0,A,u_S,u_R\rangle$, a pair $(p,\rho)$ consisting of a posterior belief distribution $p\in\Delta\Delta\Theta$ and a receiver strategy $\rho\colon\supp{p}\to\Delta A$ is an \emph{equilibrium} if
\begin{enumerate}
\item\label{eqm_p} $p$ is Bayes plausible,
\item\label{eqm_r} $\supp{(\rho(\mu))}\subseteq A(\mu)$ for all $\mu\in\supp{p}$, and
\item\label{eqm_s} $u_S(\theta)\cdot\rho(\mu)\geq u_S(\theta)\cdot\rho(\mu')$ for all $\mu, \mu' \in \supp{p}$ and all $\theta \in \supp{\mu}$.
\end{enumerate}
\end{definition}

An equilibrium $(p,\rho)$ induces the \emph{outcome} $(p,s^*)$, where $s^*\colon\Theta\to\R$ is the sender's interim expected payoff function, given by $s^*(\theta)=\max_{\mu\in\supp{p}}u_S(\theta)\cdot\rho(\mu)$.\footnote
{This follows from equilibrium condition~\ref{eqm_s} and the fact that, by the Bayes plausibility of $p$ and the full-support assumption for $\mu_0$, every state $\theta\in\Theta$ is indeed in the support of some posterior $\mu\in\supp{p}$.}
In Appendix~\ref{app:beliefbased}, we show that the outcomes that can result from any equilibrium as defined here are exactly the same as those induced by equilibria defined in the standard way for arbitrary finite sets of messages.

We say that an equilibrium $(p',\rho')$ is \emph{sender payoff-equivalent} to an equilibrium $(p,\rho)$ if both induce the same interim expected payoff function for the sender.

An important benchmark is the set of \emph{babbling} equilibria, in which the receiver infers no information at all; the posterior distribution is concentrated on the prior. Since the support of its belief distribution is the singleton $\{\mu_0\}$, any babbling equilibrium can be fully characterized by the randomized action $\rho(\mu_0)$. The only equilibrium requirement is then that $\rho(\mu_0)$ is supported on $A(\mu_0)$; the function $u_S$ is irrelevant.

\subsection{Robustness}

We are interested in the robustness of equilibria for the transparent cheap-talk games of \citet{lipnowski2020cheap} within the ambient space of general cheap-talk games.  We consider only sender payoff perturbations. From a given cheap-talk game $\Gamma=\langle \Theta, \mu_0, A, u_S, u_R \rangle$, we fix the state space $\Theta$, the prior belief $\mu_0$, the action space $A$, and the receiver preferences $u_R$; we vary only $u_S$. Let $\cU$ denote the space of all possible sender utility functions, which we identify with the Euclidean space $\R^{|A| \cdot |\Theta|}$ endowed with the usual topology.

Similarly, we can identify any equilibrium outcome $(p,s^*)$ with an element of the space $\Delta\Delta\Theta\times\R^{\abs{\Theta}}$, which we equip with the product topology, using the topology of weak convergence\footnote
{It will follow from Lemma~\ref{lem:samebeliefs} below that choosing the weak topology plays almost no role in our robustness analysis.}
on $\Delta\Delta\Theta$ and again the Euclidean topology on $\R^{\abs{\Theta}}$. Then a natural first attempt to define robustness is the following.

\begin{definition}\label{def:fullyrobust}
An equilibrium with outcome $(p,s^*)$ is \emph{fully robust} if for any pair of neighborhoods $V$ of $p$ and $W$ of $s^*$ there is a neighborhood $U$ of $u_S$ such that for any $\tilde{u}_S \in U$ there is an equilibrium whose outcome is in $V\times W$.
\end{definition}

However, full robustness turns out to be too restrictive; in cheap-talk games with transparent motives, the only payoffs that can result from fully robust equilibria are those that need no communication (see Proposition~\ref{prop:fullyrobust} in Section~\ref{sec:discussion}). Nevertheless, we still want to identify equilibria that are not knife-edge cases for specific sender utilities. Therefore, we consider open sets of perturbed utility functions \emph{nearby} the original one, but not necessarily containing it.

\begin{definition}\label{def:robust}
An equilibrium with outcome $(p,s^*)$ is \emph{robust} if for any triple of neighborhoods $U$ of $u_S$, $V$ of $p$, and $W$ of $s^*$ there is a point $u_S'\in U$ with a neighborhood $U'$ such that for any $\tilde{u}_S \in U'$ there is an equilibrium whose outcome is in $V\times W$.
\end{definition}

In effect, arbitrarily close to $u_S$, we want to have nonempty, open sets of perturbed sender utility functions $\tilde{u}_S$ that support equilibrium outcomes arbitrarily close to the given $(p,s^*)$; see Figure~\ref{fig:robusteqm}.

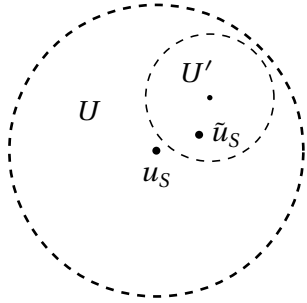
\begin{figure}[ht]
\centering
\begin{tikzpicture}
\node (O) at (0,0) {};
\node (u) at (0,-10pt) {$u_S$};
\fill (O) circle (1.5pt);
\draw[line width=1pt, dashed] (O) circle (55pt);
\node (U) at (-25pt,15pt) {$U$};
\node (u') at (20pt,20pt) {};
\fill (u') circle (1pt);
\draw[line width=.6pt, dashed] (u') circle (24pt);
\node (U') at (15pt,30pt) {$U'$};
\node (utilde) at (26pt,6pt) {$\tilde{u}_S$};
\fill (16pt,6pt) circle (1.5pt);
\end{tikzpicture}
\caption{Admissible perturbations for a robust (but not fully robust) equilibrium.}
\label{fig:robusteqm}
\end{figure}

The problem of characterizing robust equilibria is greatly simplified by the following two results. Firstly, we need not consider any other equilibrium belief distributions than the given $p$ when we perturb the sender's utility function.

\begin{lemma}\label{lem:samebeliefs}
For any cheap-talk game, the set of robust (or fully robust) equilibria does not change if, in the definition, requiring equilibrium outcomes to be in $V\times W$ is strengthened to them being in $\{p\}\times W$.
\end{lemma}

Further, we can ignore all belief distributions that are supported by more posterior beliefs than there are states. This consequence of Carath{\'e}odory's theorem still holds in our setting, which considers the sender's expected payoff at the interim and not only the ex ante stage, and, more importantly, also for robust equilibria.

\begin{lemma}\label{lem:caratheodory}
Let $\Gamma$ be a cheap-talk game and suppose $(p,\rho)$ is an equilibrium with $|\supp{p}| > |\Theta|$. Then there is another, sender payoff-equivalent, equilibrium $(p',\rho')$ in which $\supp{p'} \subseteq \supp{p}$, $|\supp{p'}| \leq |\Theta|$, and $\rho'=\rho\vert_{\supp{p'}}$. Further, if $(p,\rho)$ is (fully) robust, so is $(p',\rho')$.
\end{lemma}

Although not important for our purposes, note that the receiver will not generally be indifferent between two sender payoff-equivalent equilibria.

\subsection{Transparent games}

As mentioned above, we call, following \citet{lipnowski2020cheap}, a cheap-talk game $\Gamma$ \emph{transparent} if the sender's preferences are independent of the state, that is if $u_S(a,\theta)=u_S(a,\theta')$ for all $a\in A$ and $\theta,\theta'\in\Theta$. It follows that there is a unique function $v_S:A \to \R$ such that $u_S(a,\theta)=v_S(a)$ for all $\theta\in\Theta$ and $a\in A$.

\citet{lipnowski2020cheap} introduce the \emph{sender value correspondence}, denoted here by
\begin{equation*}
\cV:\Delta \Theta \to 2^{\R}
\end{equation*}
and given by
\begin{equation*}
\mu \mapsto \co \ v_S\left( A(\mu) \right).
\end{equation*}
By the linearity of expected payoffs, $\cV(\mu)$ contains all sender payoffs from mixed best replies of the receiver to the belief $\mu$. In a transparent cheap-talk game, the sender, whose preferences are state-independent, must expect the same payoff in all states and for all messages sent in equilibrium. Otherwise, they would strictly prefer one message over another and would do better avoiding to send the latter. Given this, \citet{lipnowski2020cheap} characterize equilibrium posterior belief distribution and sender ex ante expected payoff pairs $(p,s)$ as follows:

\begin{lemmaLR}[Lipnowski and Ravid, 2020]\label{lem:LR}
For a given transparent cheap-talk game with sender value correspondence $\cV$, there is an equilibrium inducing the posterior belief distribution $p \in \Delta \Delta \Theta$ and the ex ante expected sender payoff $s \in \R$ if and only if
\begin{enumerate}
\item the posterior belief distribution $p$ is Bayes plausible, and
\item $s\in \bigcap_{\mu\in \supp{p}}\cV(\mu)$.
\end{enumerate}
\end{lemmaLR}

The main result of \citet{lipnowski2020cheap} geometrically characterizes the sender's maximal equilibrium value. The quasiconcave envelope of the correspondence $\cV$ is the pointwise lowest quasiconcave function that majorizes $\max \cV$.\footnote
{The definition becomes slightly more complicated without the restriction to a finite state space.} 

\begin{theoremLR}[Lipnowski and Ravid, 2020]\label{thm:LR}
The sender's maximal equilibrium value (for any prior $\mu_0$) is given by the quasiconcave envelope of $\max \cV$, evaluated at $\mu_0$.
\end{theoremLR}

\section{Basic examples}\label{sec:examples}

With two states, Lemma~\ref{lem:caratheodory} shows that we can restrict ourselves to equilibria that induce only two posteriors. For such equilibria, Proposition~\ref{prop:binarybeliefrobust} below characterizes robustness completely. This characterization allows us to prove our main theorem,  Theorem~\ref{thm:smax}, which shows that every sender-optimal payoff can be robustly obtained. Crucially, in a robust equilibrium with two posteriors, one posterior must be degenerate or there must be enough best replies available to the receiver. We illustrate these two points by two simple examples. We don't solve these examples fully here; the gaps can be filled using Proposition~\ref{prop:binarybeliefrobust} below.

In Example~1, the leading example in \citet{lipnowski2020cheap}, there are two states, $\Theta=\{\theta_1,\theta_2\}$, distributed with a uniform prior $\mu_0$; the receiver has three actions available, $A=\{0,1,2\}$; the sender has a state-independent utility function given by $u_S(a,\theta)=v_S(a)=a$; and the receiver's state-dependent utility function is given by the table in Figure~\ref{fig:ex1payoffs}. Figure~\ref{fig:ex1v} illustrates the induced sender value correspondence and its quasiconcave envelope.

\begin{figure}[ht]
\centering
    \begin{tabular}{cc|cc}
      & \multicolumn{1}{c}{} & \multicolumn{2}{c}{}\\
      & \multicolumn{1}{c}{} & \multicolumn{1}{c}{$\theta_1$}  & \multicolumn{1}{c}{$\theta_2$} \\\cline{2-4}
      & $0$ & $3$ & $3$ \\
      & $1$ & $4$ & $0$ \\
       & $2$ & $0$ & $4$ \\
    \end{tabular}
    \caption{\label{fig:ex1payoffs} The receiver's payoffs in Example~1. }
\end{figure}

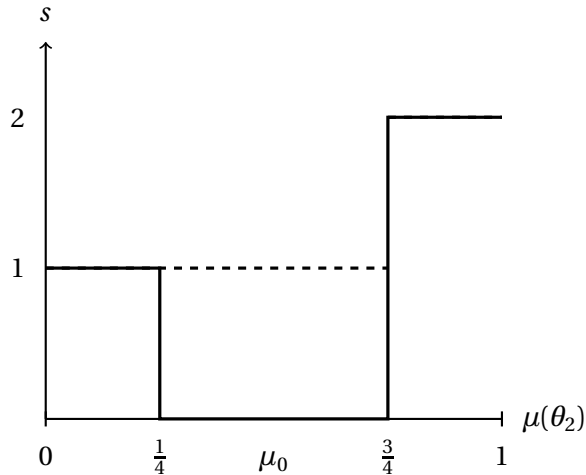
\begin{figure}[ht]
\centering
\begin{tikzpicture}
\node (0) {}
node[label=right:$\mu(\theta_2)$] (x) at (8,0) {}
node[label=below:\vphantom{$\frac14$}$0$] (x0) at (2,0) {}
node[label=below:$\frac14$] (x1) at (3.5,0) {}
node[label=below:\vphantom{$\frac14$}$\mu_0$] (x2) at (5,0) {}
node[label=below:$\frac34$] (x3) at (6.5,0) {}
node[label=below:\vphantom{$\frac14$}$1$] (x4) at (8,0) {}
node[label=left:$1$] (y1) at (2,2) {}
node[label=left:$2$] (y2) at (2,4) {}
node[label=above:$s$] (y3) at (2,5) {};
  \draw[->,thick] (2,0) -- (2,5);
  \draw[-,thick] (2,0) -- (8,0); \draw[-,thick] (8,-.1) -- (8,.1); \draw[-,thick] (2,-.1) -- (2,.1);
  \draw[very thick] (2,2) -- (3.5,2) -- (3.5,0) -- (6.5,0) -- (6.5,4) -- (8,4);
  \draw[very thick,dashed] (2,2) -- (6.5,2);
  \draw[very thick,dashed] (6.5,4) -- (8,4);
\end{tikzpicture}
\caption{\label{fig:ex1v} The sender value correspondence induced by Example~1, given by the thick black line. The dashed line is the quasiconcave envelope that characterizes the achievable sender payoffs under cheap talk \citep{lipnowski2020cheap}.}
\end{figure}

The sender's maximum expected equilibrium payoff is the value of the quasiconcave envelope at $\mu_0$; here it is $1$. By Lemma~LR, every sender payoff $s\in [0,1]$ is attainable in an equilibrium. However, only the two extreme payoffs, $s=1$ and $s=0$, are also attainable in a robust equilibrium. For $s=0$, this is straightforward; this is the payoff of a babbling equilibrium.

To see that no equilibrium  with a sender payoff in $(0,1)$ can be robust, note first that, by Lemma~LR, any such equilibrium must induce a Bayes plausible posterior belief distribution with the two posteriors $\frac14$ and $\frac34$ (of the state being $\theta_2$). A nearby equilibrium in a nearby game must attach positive probability to these two posteriors, too. This is only possible if the sender is indifferent between inducing both posteriors in both states. Since the receiver plays best replies given their beliefs, they have to play a mixture of $0$ and $1$ at posterior $\frac{1}{4}$ and a mixture of $0$ and $2$ at $\frac34$. Write the receiver's mixtures over the three actions at both beliefs as vectors in $\mathbb{R}^3$ and let the vector $x$ be the difference between the two resulting vectors. The inner product of $x$ with the payoff vector must be $0$ at both posteriors, and the entries of $x$ must sum to $0$. This gives rise to three linear equations in three unknowns. For generic payoff assignments, the system of equations has full rank and a unique solution. But one solution is $x=0$, which amounts to the receiver playing the same mixture at both states. This is then generically the unique solution. However, in this example, the receiver can only play the same mixture at both posteriors if they play $0$ for sure, which would not induce a sender payoff in $(0,1)$.

Obtaining a robust equilibrium with the sender-optimal payoff of $1$, however, is possible. This payoff can be induced using the posteriors $0$ and $\frac34$. To induce this pair of posteriors, the sender need not be indifferent in both states anymore, because these posteriors are consistent with the sender inducing \emph{always} $\frac34$ in state $\theta_2$ (and still inducing both posteriors with some probability in state $\theta_1$). Just a preference for inducing $\frac34$ in state $\theta_2$ (and indifference only in state $\theta_2$) can indeed be achieved for an open set of sender utility functions close to the original state-independent one by the receiver choosing a suitable mixture of $0$ and $2$ at $\frac34$.

Example~2 is a slight modification of Example~1. The main difference is that there are now four actions, $A=\{0,1,2,3\}$. The prior is still uniform, and the sender's state-independent utility function still satisfies $u_S(a,\theta)=v_S(a)=a$. The receiver now has payoffs given in Figure~\ref{fix:ex2payoffs}. Figure~\ref{fig:ex2v} shows the sender value correspondence and its quasiconcave envelope.

\begin{figure}[ht]
\centering
       \begin{tabular}{cc|cc}
      & \multicolumn{1}{c}{} & \multicolumn{2}{c}{}\\
      & \multicolumn{1}{c}{} & \multicolumn{1}{c}{$\theta_1$}  & \multicolumn{1}{c}{$\theta_2$} \\\cline{2-4}
      & $0$ & $3$ & $3$ \\
      & $1$ & $5$ & $-7$ \\
       & $2$ & $4$ & $0$ \\
       & $3$ & $0$ & $4$ \\
          \end{tabular}
          \caption{\label{fix:ex2payoffs} The receiver's payoffs in Example~2.}
\end{figure}

\begin{figure}[ht]
\centering
\begin{tikzpicture}
\node (0) {}
node[label=right:$\mu(\theta_2)$] (x) at (8.8,0) {}
node[label=below:\vphantom{$\frac14$}$0$] (x0) at (2,0) {}
node[label=below:$\frac18$] (x1) at (2.85,0) {}
node[label=below:$\frac14$] (x2) at (3.7,0) {}
node[label=below:\vphantom{$\frac14$}$\mu_0$] (x3) at (5.4,0) {}
node[label=below:$\frac34$] (x4) at (7.1,0) {}
node[label=below:\vphantom{$\frac14$}$1$] (x5) at (8.8,0) {}
node[label=left:$1$] (y1) at (2,1.5) {}
node[label=left:$2$] (y2) at (2,3) {}
node[label=left:$3$] (y3) at (2,4.5) {}
node[label=above:$s$] (y3) at (2,5.5) {};
  \draw[->,thick] (2,0) -- (2,5.5);
  \draw[-,thick] (2,0) -- (8.8,0); \draw[-,thick] (8.8,-.1) -- (8.8,.1); \draw[-,thick] (2,-.1) -- (2,.1);
  \draw[very thick] (2,1.5) -- (2.85,1.5) -- (2.85,3) -- (3.7,3) -- (3.7,0) -- (7.1,0) -- (7.1,4.5) -- (8.8,4.5);
  \draw[very thick,dashed] (2,1.5) -- (2.85,1.5);
  \draw[very thick,dashed] (2.85,3) -- (7.1,3);
  \draw[very thick,dashed] (7.1,4.5) -- (8.8,4.5);
\end{tikzpicture}
\caption{\label{fig:ex2v} The sender value correspondence induced by Example~2 and its quasiconcave envelope.}
\end{figure}
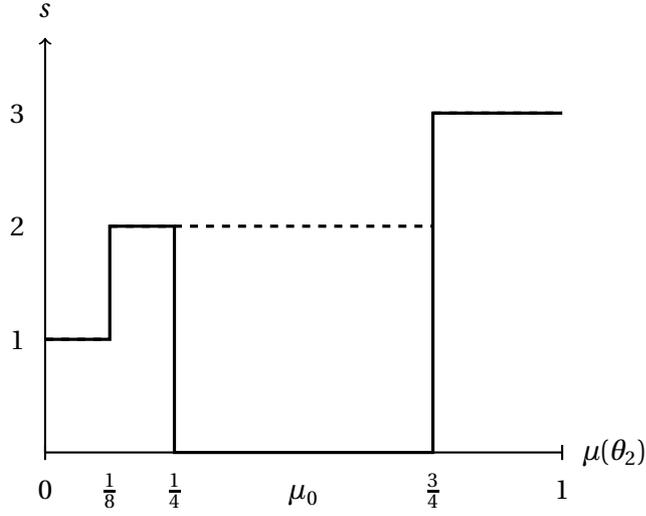

Using again Lemma~LR, every payoff in $[0,2]$ is an equilibrium payoff for the sender. The payoff $0$ can be robustly obtained using a babbling equilibrium. A payoff in $(0,1)$ cannot be robustly obtained; the reason is the same as in the last example. But now, every sender payoff in $[1,2]$ can be robustly obtained by inducing the two posteriors $\frac 18$ and $\frac34$. Now, the relevant mixtures are over four actions. Forming the same system of linear equations gives us now three equations in four unknowns. Generically, we get a one-dimensional space of solutions, which turns out to be enough for satisfying all relevant incentive constraints.

\section{Main result}\label{sec:results}

\begin{theorem}\label{thm:smax}
Let $\Gamma$ be a transparent cheap-talk game with binary state space. Then the sender's maximal equilibrium value is attainable in a robust equilibrium.
\end{theorem}

The proof of Theorem~\ref{thm:smax} ultimately relies on a number of auxiliary results. In this section, we present the central underlying results that effectively provide a characterization of all sender payoffs that are attainable in a robust equilibrium (if the state space is binary). From this characterization, we will then derive a more graphical one in Section~\ref{sec:graphical}, and that will finally enable us to give a short proof of Theorem~\ref{thm:smax}. 

To begin, we first do away with the trivial case of a \emph{babbling} equilibrium, which means that with probability one the receiver sticks to the prior belief. 

\begin{proposition}\label{prop:babbling}
Any babbling equilibrium is fully robust.
\end{proposition}

\begin{proof}
Let $\Gamma$ be a cheap-talk game and $(p,\rho)$ an equilibrium with $p(\mu_0)=1$. Then the same equilibrium persists for \emph{any} perturbed sender utility function $\tilde{u}_S$ by the simple fact that there is no belief $\mu \neq \mu_0$ in the support of $p$. The sender's interim expected payoff in this equilibrium for the perturbed game then is $\tilde{s}^*$ given by $\tilde{s}^*(\theta)=\tilde{u}_S(\theta)\cdot\rho(\mu_0)$, which depends continuously on $\tilde{u}_S$.
\end{proof}

By Lemma~LR, the sender payoffs in a transparent cheap-talk game that can be sustained by a babbling equilibrium are all $s\in\cV(\mu_0)$, so we now turn to payoffs $s>\max\cV(\mu_0)$.

Assume, from now on, that the state space is indeed binary. Then, by Proposition~\ref{prop:babbling} and Lemma~\ref{lem:caratheodory}, we only need to characterize the robust equilibria with exactly two posteriors in the support of the belief distribution. We achieve this in the following proposition. To simplify notation, let specifically $\Theta=\{\theta_1,\theta_2\}$; this allows us to identify any belief $\mu\in\Delta\Theta$ by the probability it assigns to $\theta_2$, so that we can treat $\mu$ as a number in $[0,1]$.

\begin{proposition}\label{prop:binarybeliefrobust}
Let $\Gamma$ be a transparent cheap-talk game with binary state space. An equilibrium with an expected sender payoff $s>\max\cV(\mu_0)$ and a belief distribution $p$ having binary support is robust if and only if one of the following holds:
\begin{enumerate}[label=\arabic*.]
\item\label{fullyinf}
$\supp{p}=\{0,1\}$.
\item\label{threeactions}
$\supp{p}$ contains a posterior $\mu\in\{0,1\}$ and $\abs{\bigcup_{\mu\in\supp{p}}A(\mu)}\geq 3$.
\item\label{fouractions}
$\abs{\bigcup_{\mu\in\supp{p}}A(\mu)}\geq 4$.
\end{enumerate}
\end{proposition}

Proposition~\ref{prop:binarybeliefrobust} is the main characterization result, but its proof (given in Appendix~\ref{app:binarybeliefrobust}) involves a number of steps and technical results, so we provide only a sketch here. Recall that we can fix the belief distribution $p$ thanks to Lemma~\ref{lem:samebeliefs}.

The easier part of the equivalence is necessity. Condition~\ref{fullyinf} means that there is no further restriction if the equilibrium is fully informative. So suppose $p$ involves a posterior belief that is supported on both states. Then the sender must be indifferent in some state. There are three ways to achieve this indifference. First, the receiver might use the same mixed action for both posteriors. In this case, it is possible to show that there exists also a babbling equilibrium with the same expected sender payoff, which is in conflict with $s>\max\cV(\mu_0)$. Second, the sender might derive the same utility from different actions that the receiver uses. Such sender utility functions are, however, ``nongeneric'' and contained in a lower-dimensional space. Third, the receiver might have enough degrees of freedom to choose different mixed actions for the two posteriors (such that the corresponding ranges of sender payoffs overlap). This amounts to three optimal actions for the receiver when the sender needs to be indifferent in only one state (condition~\ref{threeactions}), and to four optimal actions when indifference must hold in both states (condition~\ref{fouractions}).

To see the latter in more detail, and to prepare for sketching also the proof of sufficiency, let the two posteriors in the support of $p$ be $\mu$ and $\mu'$. Then any equilibrium $(p,\tilde{\rho})$ for a perturbed sender utility function $\tilde{u}_S$ can be characterized in terms of the difference $x=\tilde{\rho}(\mu)-\tilde{\rho}(\mu')$. In particular, the set of all such $x$ where $\tilde{\rho}$ satisfies equilibrium condition~\ref{eqm_r} is
\begin{equation*}
\cD:=\setd{r-r'}{\text{$r,r'\in\Delta A$, $\supp{r}\subseteq A(\mu)$, and $\supp{r'}\subseteq A(\mu')$}}.
\end{equation*}
Further, the sender is indifferent in state $\theta$ if and only if $\tilde{u}_S(\theta)\cdot x=0$. We, thus, need a nonnull $x\in\cD$ that solves the indifference equation(s). This requires $\cD$ to be big enough, and the dimension of the linear span of $\cD$ is precisely $\abs{A(\mu)\cup A(\mu')}-1$ (see Lemma~\ref{lem:dimD} in Appendix~\ref{app:D}).

The proof of sufficiency consists of two steps and uses the same representation of equilibria in terms of the vectors $x$. But now we also need to take care of the sender's (weak) preference if some state is revealed to the receiver, i.e., if some posterior belief is supported by only one state $\theta$. Then equilibrium condition~\ref{eqm_s} corresponds to an inequality of the form $\tilde{u}_S(\theta)\cdot x\geq 0$ or $\tilde{u}_S(\theta)\cdot x\leq 0$. Hence, letting
\begin{equation*}
H_{x}=\setd{y\in\R^n}{y\cdot x=0},
\end{equation*}
\begin{equation*}
H_{x}^{+}=\setd{y\in\R^n}{y\cdot x\geq 0},
\end{equation*}
and 
\begin{equation*}
H_{x}^{-}=\setd{y\in\R^n}{y\cdot x\leq 0}
\end{equation*}
for any $x\in\R^n$, the set of utility functions $\tilde{u}_S$ for which there is an equilibrium $(p,\tilde{\rho})$ is 
\begin{equation*}
\bigcup_{x\in\cD}H_{x}^{+}\times H_{x}^{-}
\end{equation*}
if $\mu=0$ and $\mu'=1$,
\begin{equation*}
\bigcup_{x\in\cD}H_{x}\times H_{x}^{-}
\end{equation*}
if $\mu=0$ and $\mu'<1$,
\begin{equation*}
\bigcup_{x\in\cD}H_{x}^{+}\times H_{x}
\end{equation*}
if $\mu>0$ and $\mu'=1$, and
\begin{equation*}
\bigcup_{x\in\cD}H_{x}\times H_{x}
\end{equation*}
if $\mu>0$ and $\mu'<1$. In the first step of the proof, we address each of these four cases and show that, if $\cD$ is big enough, the given set of utility functions has nonempty interior, that the interior actually has a nonempty intersection with any neighborhood of the state-independent utility function $u_S$, and that the same is still true if we consider only those $x\in\cD$ that are in an arbitrary neighborhood of $\rho(\mu)-\rho(\mu')$ (where $\rho$ is from the given equilibrium for $u_S$). Most of the technical work is necessary for the case of one equality paired with an inequality and the case of two equalities, so the specific arguments for the two cases are presented as separate propositions (in Appendix~\ref{app:binarybeliefrobust}).

In the second step, we then show that any $x\in\cD$ sufficiently close to $\rho(\mu)-\rho(\mu')$ has in fact a representation $x=\tilde{\rho}(\mu)-\tilde{\rho}(\mu')$ with $\tilde{\rho}(\mu)$ arbitrarily close to $\rho(\mu)$ and $\tilde{\rho}(\mu')$ arbitrarily close to $\rho(\mu')$ (see Proposition~\ref{prop:decomp} in Appendix~\ref{app:binarybeliefrobust}). This implies that we can make the interim expected sender payoff in the equilibrium $(p,\tilde{\rho})$ as close as desired to the one in $(p,\rho)$, which finally yields robustness.

Combining Proposition~\ref{prop:babbling}, Lemma~\ref{lem:caratheodory}, and Proposition~\ref{prop:binarybeliefrobust}, we now have a full characterization of all sender payoffs that are attainable in a robust equilibrium. However, this characterization is not very convenient, yet, for proving Theorem~\ref{thm:smax}.

\section{A graphical characterization}\label{sec:graphical}

%
%
%

From Proposition~\ref{prop:binarybeliefrobust}, we can in fact derive an alternative characterization of all sender payoffs that are attainable in a robust equilibrium (other than a babbling equilibrium, which is always robust by Proposition~\ref{prop:babbling}). This new characterization uses only the prior $\mu_0$ and the sender payoff $s$ but not a specific equilibrium. It is particularly useful for a visual analysis of the graph of the sender value correspondence when the receiver's optimal actions change more often than in the examples given in Section~\ref{sec:examples}. But it will also make a short proof of Theorem~\ref{thm:smax} possible (see Subsection~\ref{sec:proofsmax}). 

Recall that, by the Quasiconcavification Theorem of \citet{lipnowski2020cheap}, the maximal equilibrium payoff for the sender is $s=c(\mu_0)$, where $c$ denotes the quasiconcave envelope of the sender value correspondence $\cV$.

\begin{proposition}\label{prop:srobust}
Let $\Gamma$ be a transparent cheap-talk game with binary state space. Assume $c(\mu_0)>\max\cV(\mu_0)$ and consider any sender payoff $s$ such that $c(\mu_0)\geq s>\max\cV(\mu_0)$. Then there are two beliefs $\mu,\mu'\in\Delta\Theta$ satisfying
\begin{equation}\label{extremesupport}
\mu=\min\setd{\tilde{\mu}\in\Delta\Theta}{\text{$\tilde{\mu}<\mu_0$ and $s\in\cV(\tilde{\mu})$}} \quad \text{and} \quad \mu'=\max\setd{\tilde{\mu}\in\Delta\Theta}{\text{$\tilde{\mu}>\mu_0$ and $s\in\cV(\tilde{\mu})$}}.
\end{equation}
Fix these two beliefs. There exists a robust equilibrium attaining $s$ if and only if one of the following holds:
\begin{enumerate}[label=\alph*)]
\item\label{a)} $s\in\cV(0)$ or $s\in\cV(1)$.
\item\label{b)} There exists a belief $\tilde{\mu}\in(\mu,\mu')$ such that $A(\tilde{\mu})\neq A(\mu_0)$.
\item\label{c)} $\abs{A(\mu)}\geq 3$ or $\abs{A(\mu')}\geq 3$.
\end{enumerate}
A sufficient condition for $s$ to be attainable in a robust equilibrium is that there are at least two actions that have utility less than or equal to $s$ for the sender and that are, respectively, optimal for the receiver at some belief in $\Delta\Theta$.
\end{proposition}

The two beliefs $\mu$ and $\mu'$ are the extreme beliefs that can support an equilibrium with sender payoff $s$; they are the respectively left- and rightmost points where $\cV$ intersects the constant payoff $s$. Condition~\ref{a)} means that at least one of them is degenerate, and condition~\ref{b)} means that (strictly) between them the set of optimal actions for the receiver is not constant. 

The further sufficient condition, which does not involve $\mu$ and $\mu'$, cannot replace any of the three other conditions in terms of necessity. However, we obtain a further significant simplification for generic games.\footnote
{Here we mean generic within the space of transparent games, which are themselves nongeneric within the ambient space of nontransparent games.} 
For a generic receiver utility function, there are at most only two optimal actions at any given belief. Then condition~\ref{c)} becomes irrelevant. Moreover, in the proof of Proposition~\ref{prop:srobust} we show that condition~\ref{b)} is equivalent to $A(\mu)\cap A(\mu')=\emptyset$. Since $s$ is in both $\cV(\mu)$ and $\cV(\mu')$, condition~\ref{b)} thus implies the further sufficient condition. Assuming additionally that the sender is never indifferent between any two actions, we then conclude that $s$ is attainable in a robust equilibrium if and only if $s$ is in $\cV(0)\cup\cV(1)$ or at least equal to the second lowest value that $\cV$ ever attains.

\subsection{Applying the graphical characterization}\label{sec:graphicalex}

In this subsection, we illustrate Proposition~\ref{prop:srobust} by further examples. Since we only sketch the respective sender value correspondences (in Figures \ref{fig:ex3v} and \ref{fig:ex4v}), assume the underlying games are generic in the sense specified just before. Thus, we can identify the pairs $(\mu_0,s)$ such that $s$ is attainable in a robust equilibrium also by the simpler criterion. The sets of all these pairs form the colored rectangles.

\definecolor{fill1}{rgb}{.8,.9,1}
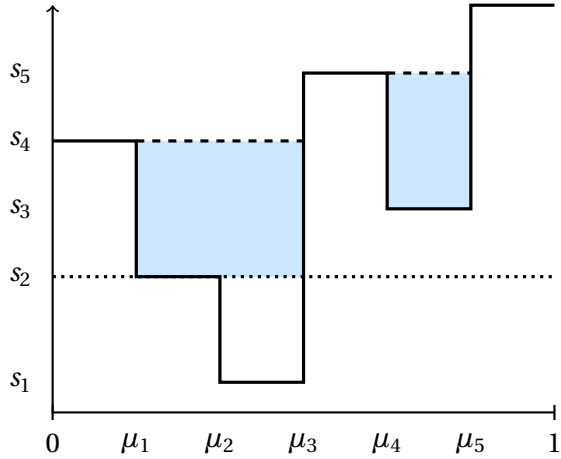
\begin{figure}[ht]
\centering
\begin{tikzpicture}
\node (0) {}
node[label=below:$0$] (x0) at (0,0) {}
node[label=below:$\mu_1$] (x1) at (1.1,0) {}
node[label=below:$\mu_2$] (x2) at (2.2,0) {}
node[label=below:$\mu_3$] (x3) at (3.3,0) {}
node[label=below:$\mu_4$] (x4) at (4.4,0) {}
node[label=below:$\mu_5$] (x5) at (5.5,0) {}
node[label=below:$1$] (x4) at (6.6,0) {}
node[label=left:$s_1$] (y1) at (0,.4) {}
node[label=left:$s_2$] (y2) at (0,1.8) {}
node[label=left:$s_3$] (y3) at (0,2.7) {}
node[label=left:$s_4$] (y4) at (0,3.6) {}
node[label=left:$s_5$] (y5) at (0,4.5) {};
  \draw[->,thick] (0,0) -- (0,5.4);
  \draw[-,thick] (0,0) -- (6.6,0); \draw[-,thick] (6.6,-.1) -- (6.6,.1); \draw[-,thick] (0,-.1) -- (0,.1);
  \fill[fill1] (1.1,3.6) rectangle (3.3,1.8);
  \fill[fill1] (4.4,4.5) rectangle (5.5,2.7);
  \draw[very thick] (0,3.6) -- (1.1,3.6) -- (1.1,1.8) -- (2.2,1.8) -- (2.2,.4) -- (3.3,.4) -- (3.3,4.5) -- (4.4,4.5) -- (4.4,2.7) -- (5.5,2.7) -- (5.5,5.4) -- (6.6,5.4);
  \draw[very thick,dashed] (1.1,3.6) -- (3.3,3.6);
  \draw[very thick,dashed] (4.4,4.5) -- (5.5,4.5);
  \draw[very thick,dotted] (0,1.8) -- (6.6,1.8);
\end{tikzpicture}
\caption{\label{fig:ex3v} Example~3.}
\end{figure}

In Example~3, shown in Figure~\ref{fig:ex3v}, we obtain all robustly attainable sender payoffs from the condition that $s$ is at least equal to the second lowest value that $\cV$ attains, which is $s_2$. The condition that $s$ is in $\cV(0)\cup\cV(1)$ does not add anything, because then also $s\geq s_2$.

A direct application of Proposition~\ref{prop:srobust} yields the following. In this example, there are five different pairs $\mu$, $\mu'$ that result from different choices of $s$ and that, then, apply to any prior $\mu_0\in(\mu,\mu')$ such that $c(\mu_0)\geq s>\max\cV(\mu_0)$:

\begin{enumerate}[label=\arabic*.]
\item For any $s\in(s_1,s_2)$, we have $\mu=\mu_2$ and $\mu'=\mu_3$. This case is analogous to $s\in(0,1)$ in Examples 1 and 2 (cf.\ Figures \ref{fig:ex1v} and \ref{fig:ex2v}), so these sender payoffs are not attainable in any robust equilibrium for the reasons given in Section~\ref{sec:examples}. Correspondingly, neither condition~\ref{a)} nor condition~\ref{b)} is satisfied here: $s$ is not in $\cV(0)$ or $\cV(1)$, and the set of optimal actions for the receiver is constant for all posterior beliefs in $(\mu,\mu')$ (since different actions would imply different sender payoffs in a generic game).

\item For any $s\in[s_2,s_3)$, we now have $\mu=\mu_1$ but still $\mu'=\mu_3$. Then condition~\ref{b)} holds, because the sender payoff correspondence is not constant in this interval and, thus, also the set of optimal actions for the receiver must undergo a change at $\mu_2$.

\item For any $s\in[s_3,s_4)$, we still have $\mu=\mu_1$ but $\mu'$ increases to $\mu_5$. This means that condition~\ref{b)} is inherited from the previous case.

\item The case $s=s_4$ is special, because then $\mu=0$, respectively $s\in\cV(0)$. This is the only instance in this example in which actually condition~\ref{a)} holds. But since still $\mu'=\mu_5$, also condition~\ref{b)} holds.

\item For any $s\in(s_4,s_5]$, we finally have $\mu=\mu_3$ and $\mu'=\mu_5$. Then condition~\ref{b)} holds because the set of optimal actions for the receiver changes at $\mu_4$.
\end{enumerate}

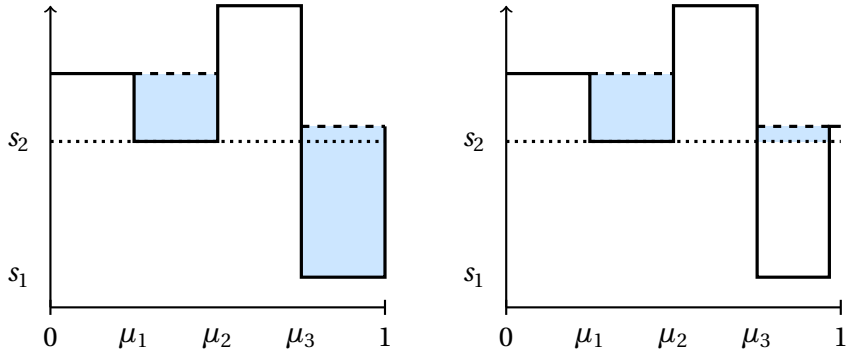
\begin{figure}[ht]
\centering
\begin{tikzpicture}
\node (0) {}
node[label=below:$0$] (x0) at (0,.5) {}
node[label=below:$\mu_1$] (x1) at (1.1,.5) {}
node[label=below:$\mu_2$] (x2) at (2.2,.5) {}
node[label=below:$\mu_3$] (x3) at (3.3,.5) {}
node[label=below:$1$] (x4) at (4.4,.5) {}
node[label=left:$s_1$] (y1) at (0,.9) {}
node[label=left:$s_2$] (y2) at (0,2.7) {};
  \draw[->,thick] (0,.5) -- (0,4.5);
  \draw[-,thick] (0,.5) -- (4.4,.5); \draw[-,thick] (4.4,.4) -- (4.4,.6); \draw[-,thick] (0,.4) -- (0,.6);
  \fill[fill1] (3.3,2.9) rectangle (4.4,.9);
  \fill[fill1] (1.1,3.6) rectangle (2.2,2.7);
  \draw[very thick] (0,3.6) --  (1.1,3.6) -- (1.1,2.7) -- (2.2,2.7) -- (2.2,4.5) -- (3.3,4.5) -- (3.3,.9) -- (4.4,.9) -- (4.4,2.9);
  \draw[very thick,dashed] (3.3,2.9) -- (4.4,2.9);
  \draw[very thick,dashed] (1.1,3.6) -- (2.2,3.6);
  \draw[very thick,dotted] (0,2.7) -- (4.4,2.7);
\end{tikzpicture}
\hspace{.5cm}
\begin{tikzpicture}
\node (0) {}
node[label=below:$0$] (x0) at (0,.5) {}
node[label=below:$\mu_1$] (x1) at (1.1,.5) {}
node[label=below:$\mu_2$] (x2) at (2.2,.5) {}
node[label=below:$\mu_3$] (x3) at (3.3,.5) {}
node[label=below:$1$] (x4) at (4.4,.5) {}
node[label=left:$s_1$] (y1) at (0,.9) {}
node[label=left:$s_2$] (y2) at (0,2.7) {};
  \draw[->,thick] (0,.5) -- (0,4.5);
  \draw[-,thick] (0,.5) -- (4.4,.5); \draw[-,thick] (4.4,.4) -- (4.4,.6); \draw[-,thick] (0,0.4) -- (0,.6);
  \fill[fill1] (3.3,2.9) rectangle (4.25,2.7);
  \fill[fill1] (1.1,3.6) rectangle (2.2,2.7);
  \draw[very thick] (0,3.6) --  (1.1,3.6) -- (1.1,2.7) -- (2.2,2.7) -- (2.2,4.5) -- (3.3,4.5) -- (3.3,.9) -- (4.25,.9) -- (4.25,2.9) -- (4.4,2.9);
  \draw[very thick,dashed] (3.3,2.9) -- (4.25,2.9);
  \draw[very thick,dashed] (1.1,3.6) -- (2.2,3.6);
  \draw[very thick,dotted] (0,2.7) -- (4.4,2.7);
\end{tikzpicture}
\caption{\label{fig:ex4v} Examples 4a and 4b.}
\end{figure}

Next, in Example~4a, shown in the left panel of Figure~\ref{fig:ex4v}, condition~\ref{a)} makes more sender equilibrium payoffs robustly attainable, because here $(s_1,s_2)\subseteq\cV(1)$. In consequence, \emph{every} sender equilibrium payoff is attainable also in a robust equilibrium. However, condition~\ref{a)} becomes void again if the receiver has a unique optimal action at the extreme belief $1$. If the unique optimal action is the one with sender payoff $s_1$, then still all sender equilibrium payoffs are robustly attainable, but now simply because the only equilibrium payoff for any $\mu_0>\mu_3$ would be the babbling equilibrium payoff $s_1$. If, instead, the unique optimal action is the other one with payoff slightly greater than $s_2$, then we arrive at Example~4b, shown in the right panel of Figure~\ref{fig:ex4v}, and the equilibrium payoffs $s\in(s_1,s_2)$ for any prior such that $\cV(\mu_0)=\{s_1\}$ are not robustly attainable anymore.

\subsection{Proof of Theorem~\ref{thm:smax}}\label{sec:proofsmax}

Finally, Proposition~\ref{prop:srobust} facilitates a short proof of Theorem~\ref{thm:smax} (also for nongeneric games):

If $c(\mu_0)=\max\cV(\mu_0)$, there is a babbling equilibrium with payoff $s=c(\mu_0)$, which is robust by Proposition~\ref{prop:babbling}. Hence, suppose $c(\mu_0)>\max\cV(\mu_0)$ and consider $s=c(\mu_0)$. We are going to argue that the further sufficient condition in Proposition~\ref{prop:srobust} holds whenever condition~\ref{b)} is not satisfied.

Since $s>\max\cV(\mu_0)$, all actions in $A(\mu_0)$ yield a sender payoff of at most $s$, so there is at least one such action. Moreover, in the proof of Proposition~\ref{prop:srobust} we show that if condition~\ref{b)} does not hold, then there is some $a\in A(\mu)\setminus A(\mu_0)$ and some $a'\in A(\mu')\setminus A(\mu_0)$. But since $s=c(\mu_0)$, we must have $s\geq\sup\bigcup_{\tilde{\mu}\leq\mu_0}\cV(\tilde{\mu})$ or $s\geq\sup\bigcup_{\tilde{\mu}\geq\mu_0}\cV(\tilde{\mu})$.\footnote
{In fact, since $A$ is finite, $c(\mu_0)=\min(\max\bigcup_{\tilde{\mu}\leq\mu_0}\cV(\tilde{\mu}),\max\bigcup_{\tilde{\mu}\geq\mu_0}\cV(\tilde{\mu}))$.}
It follows that also $v_S(a)\leq s$ or $v_S(a')\leq s$.
\qed

\section{Discussion}\label{sec:discussion} 

We showed that in transparent cheap-talk games with two states and finitely many actions, every sender-optimal equilibrium payoff can be approximated by equilibrium payoffs of open sets of nearby games that need not be transparent anymore. The belief distributions and interim payoffs of the approximating equilibria are close to a sender-optimal equilibrium in the original transparent game, which we then call a robust equilibrium. At the heart of our argument is Proposition~\ref{prop:binarybeliefrobust}, which characterizes a canonical set of robust equilibria spanning all outcomes induced by any robust equilibrium. Lemma~\ref{lem:samebeliefs} shows that one can use the belief distribution of this equilibrium as the belief distribution of the equilibria of the approximating games, and Lemma~\ref{lem:caratheodory} shows that we only have to look among equilibria that induce at most two posteriors. These results lead rather easily to some necessary conditions a candidate equilibrium must satisfy, and the bulk of the work consists of showing that these necessary conditions suffice.

What happens if we have more than two states? While we were not able to find a counterexample in which the sender-optimal payoff cannot be achieved robustly (we looked at a number of games with three states), our approach does not readily provide a path to a proof that it always can. With three or more states, the space of equilibria inducing an outcome becomes much larger, the space of beliefs becomes multidimensional, and there is no obvious candidate robust equilibrium to work with. Finding the right set of necessary conditions becomes much more difficult, let alone proving their sufficiency.

Even if one ignores all robustness concerns, the geometric structure of sender-optimal payoffs becomes much more complex with more than two states. With only two states, using the quasiconcavification approach of \citet{lipnowski2020cheap} is straightforward: Fill up all valleys and look at the skyline. No such simple approach for finding the quasiconcavification exists with more states, and the geometry of the problem becomes much more subtle.

It is also worth discussing our robustness notion more closely. One way to think about transparent games is as representing games that are not completely transparent, but close enough. A robust equilibrium of a transparent game can then be approximated by equilibria of an open set of such games. A more straightforward approach would be to require this open set to include the initial transparent game. This would give us what we called fully robust equilibria. By Proposition~\ref{prop:babbling}, any babbling equilibrium is fully robust. This is essentially also as far as one can go: Every sender payoff in a fully robust equilibrium could be achieved using a babbling equilibrium. 

\begin{proposition}\label{prop:fullyrobust}
Let $\Gamma$ be a transparent cheap-talk game with binary state space. Then, for any fully robust equilibrium, there is a babbling equilibrium that yields the same payoff to the sender.
\end{proposition}

Proposition~\ref{prop:fullyrobust} also shows that our robustness notion cannot be meaningfully strengthened while still covering nontrivial equilibria. It implies that for any nontrivial equilibrium of a transparent game, the initial transparent game is the limit of a sequence of games all of whose equilibria avoid a closed neighborhood of the given nontrivial equilibrium. Standard arguments show that the equilibrium correspondence is upper hemicontinuous. Consequently, one can approximate each of the approximating games, and thus also the transparent limit game, by open sets of games whose equilibria avoid the neighborhood.

\appendix

\section{Auxiliary results}\label{app:aux}

\subsection{Belief-based approach}\label{app:beliefbased}

Given any finite set of messages $M$, recall that an \emph{equilibrium} for a cheap-talk game consists of a sender strategy $\sigma\colon\Theta\to\Delta M$, a receiver strategy $\rho\colon M\to\Delta A$, and a belief system $\beta\colon M\to\Delta\Theta$ such that
\begin{enumerate}[label=\arabic*.]
\item\label{eqmbelief}
$\beta$ is obtained from $\mu_0$, given $\sigma$, using Bayes' rule, which means that
\begin{equation}\label{Bayes}
\sigma(m\vert\theta)\mu_0(\theta)=\beta(\theta\vert m)\sum_{\theta'\in\Theta}\sigma(m\vert\theta')\mu_0(\theta')\quad\text{for all $m\in M$ and $\theta\in\Theta$,}
\end{equation}
\item\label{eqmRstrat}
$\rho(m)$ is supported on $A(\beta(m))$ for all $m\in M$, and
\item\label{eqmSstrat}
$\sigma(\theta)$ is supported on $\arg\max_{m\in M}u_S(\theta)\cdot\rho(m)$ for all $\theta\in\Theta$.
\end{enumerate}

However, instead of working with explicit message sets and sender strategies, we adopt the so-called \emph{belief-based approach}. Since equilibrium payoffs for the sender may generally depend on the state $\theta$ in our case, we use the following characterization of equilibria.

\begin{lemma}\label{lem:beliefbased}
Let $p\in\Delta\Delta\Theta$ be a finite support distribution over beliefs, and let $s^*\colon\Theta\to\R$ be a function that maps states to sender payoffs. Then there exists an equilibrium $(\sigma,\rho,\beta)$ with a finite message set $M$ such that
\begin{equation}\label{beliefdist}
p(\mu)=\sum_{\substack{\theta\in\Theta \\ m\in\beta^{-1}(\mu)}}\sigma(m\vert\theta)\mu_0(\theta)
\end{equation}
for all $\mu\in\Delta\Theta$ and
\begin{equation}\label{interimSvalue}
s^*(\theta)=\max_{m\in M}u_S(\theta)\cdot\rho(m)
\end{equation}
for all $\theta\in\Theta$ if and only if there exists a function $\rho'\colon\supp{p}\to\Delta A$ such that
\begin{enumerate}
\item\label{Bayesplausible}
$\sum_{\mu\in\supp{p}}\mu p(\mu)=\mu_0$,
\item\label{ICR}
$\supp{(\rho'(\mu))}\subseteq A(\mu)$ for all $\mu\in\supp{p}$, and
\item\label{ICS}
$s^*(\theta)=u_S(\theta)\cdot\rho'(\mu)\geq u_S(\theta)\cdot\rho'(\mu')$ for all $\mu,\mu'\in\supp{p}$ and $\theta\in\supp{\mu}$.
\end{enumerate}
\end{lemma}

\begin{proof}
``$\Rightarrow$'' Let $\supp{p}=\{\mu_1,\dots,\mu_K\}$ and suppose $(\sigma,\rho,\beta)$ is an equilibrium with finite message set $M$ such that \eqref{beliefdist} holds for all $\mu\in\Delta\Theta$. Then, by \eqref{Bayes}, $\sum_{m\in\beta^{-1}(\mu_k)}\sigma(m\vert\theta)\mu_0(\theta)=\mu_k(\theta)p(\mu_k)$ for all $k=1,\dots,K$ and $\theta\in\Theta$. Thus, condition~\ref{Bayesplausible} holds if $\sum_{k=1}^K\sum_{m\in\beta^{-1}(\mu_k)}\sigma(m\vert\theta)=1$ for all $\theta$, and the latter actually follows from the facts that $1=\sum_{k=1}^Kp(\mu_k)=\sum_{k=1}^K\sum_{\theta\in\Theta,m\in\beta^{-1}(\mu_k)}\sigma(m\vert\theta)\mu_0(\theta)$, $\sum_{\theta\in\Theta}\mu_0(\theta)=1$, and $\sum_{k=1}^K\sum_{m\in\beta^{-1}(\mu_k)}\sigma(m\vert\theta)\leq 1$ for all $\theta$.

Now suppose additionally \eqref{interimSvalue} holds for all $\theta\in\Theta$. Consider any $\theta\in\Theta$ and $k=1,\dots,K$ such that $\theta\in\supp{\mu_k}$, i.e., $\mu_k(\theta)>0$. Then, for any $m\in\beta^{-1}(\mu_k)$, \eqref{Bayes} and $\mu_0(\theta)>0$ together imply that $\sigma(m\vert\theta)>0$ if and only if $\sum_{\theta'\in\Theta}\sigma(m\vert\theta')\mu_0(\theta')>0$. Thus, by equilibrium condition~\ref{eqmSstrat}, any $m_k\in\beta^{-1}(\mu_k)$ with $\sum_{\theta'\in\Theta}\sigma(m_k\vert\theta')\mu_0(\theta')>0$ must satisfy $u_S(\theta)\cdot\rho(m_k)=s^*(\theta)$. Such $m_k$ indeed exists because of $\sum_{\theta'\in\Theta,m\in\beta^{-1}(\mu_k)}\sigma(m\vert\theta')\mu_0(\theta')=p(\mu_k)>0$. Fix one such $m_k$ for each $k=1,\dots,K$ and define the function $\rho'\colon\{\mu_1,\dots,\mu_K\}\to\Delta A$ by $\rho'(\mu_k)=\rho(m_k)$. Then condition~\ref{ICR} holds by construction and equilibrium condition~\ref{eqmRstrat}, since $\beta(m_k)=\mu_k$. Moreover, by construction of $m_k$ and $\rho'$, whenever $\theta\in\supp{\mu_k}$, then $u_S(\theta)\cdot\rho'(\mu_k)=s^*(\theta)\geq u_S(\theta)\cdot\rho'(\mu_j)$ for all $j=1,\dots,K$, so also condition~\ref{ICS} holds.

``$\Leftarrow$'' Suppose there is a function $\rho'\colon\supp{p}\to\Delta A$ such that the conditions \ref{Bayesplausible}, \ref{ICR}, and \ref{ICS} hold. Let $M:=\supp{p}=\{\mu_1,\dots,\mu_K\}$, define the belief system $\beta\colon M\to\Delta\Theta$ by $\beta(\mu_k)=\mu_k$, and fix the receiver strategy $\rho:=\rho'$. Finally, construct the sender strategy $\sigma\colon\Theta\to\Delta M$ by $\sigma(\mu_k\vert\theta)\mu_0(\theta)=\mu_k(\theta)p(\mu_k)$, which is possible because the latter equation implies $\sigma(\mu_k\vert\theta)>0$ and, by condition~\ref{Bayesplausible}, $\sum_{k=1}^K\sigma(\mu_k\vert\theta)=1$. Further, then $\sum_{\theta\in\Theta}\sigma(\mu_k\vert\theta)\mu_0(\theta)=p(\mu_k)$ by $\mu_k\in\Delta\Theta$, which implies that \eqref{Bayes} holds. Equilibrium condition~\ref{eqmRstrat} holds by construction, since $\supp{(\rho(\mu_k))}\subseteq A(\mu_k)=A(\beta(\mu_k))$ for all $\mu_k\in M$. To verify also equilibrium condition~\ref{eqmSstrat} and that \eqref{interimSvalue} holds for all $\theta\in\Theta$, suppose $\sigma(\mu_k\vert\theta)>0$. Then by construction $\mu_k(\theta)>0$, i.e., $\theta\in\supp{\mu_k}$, so that condition~\ref{ICS} indeed implies $s^*(\theta)=u_S(\theta)\cdot\rho(\mu_k)\geq u_S(\theta)\cdot\rho(\mu_j)$ for all $\mu_j\in M$.
\end{proof}

\subsection{Belief dependence of the receiver's optimal actions}\label{app:A}

Here we study how the belief-dependent sets of optimal actions for the receiver, $A(\mu)$, are related to each other for different beliefs $\mu\in\Delta\Theta$.

\begin{lemma}\label{lem:Aneighborhood}
For every belief $\mu\in\Delta\Theta$ there exists a neighborhood $U_{\mu}$ such that $A(\tilde{\mu})\subseteq A(\mu)$ for all $\tilde{\mu}\in U_{\mu}$.
\end{lemma}

\begin{proof}
This follows from upper hemicontinuity (by Berge's Maximum Theorem) of the correspondence that associates to each $\mu\in\Delta\Theta$ the set $A(\mu)$ and finiteness of $A$.
\end{proof}

\begin{lemma}\label{lem:Aintersection}
Let $\mu,\mu'\in\Delta\Theta$ be two given beliefs. If $A(\mu)$ and $A(\mu')$ have any common element, then, for all $\tilde{\mu}\in\co\{\mu,\mu'\}$, $A(\tilde{\mu})=A(\mu)\cap A(\mu')$ or $\tilde{\mu}\in\{\mu,\mu'\}$.
\end{lemma}

\begin{proof}
Suppose $a\in A(\mu)\cap A(\mu')$ and $\tilde{\mu}=\alpha\mu+(1-\alpha)\mu'$ for some $\alpha\in(0,1)$. For notational simplicity, let $u_R(a')\cdot\mu$ denote $\sum_{\theta\in\Theta}u_R(a',\theta)\mu(\theta)$ for any $a'\in A$. Proceeding likewise also for other beliefs, $(u_R(a)-u_R(a'))\cdot\tilde{\mu}=\alpha(u_R(a)-u_R(a'))\cdot\mu+(1-\alpha)(u_R(a)-u_R(a'))\cdot\mu'$, which is nonnegative by $a\in A(\mu)\cap A(\mu')$. Thus, also $a\in A(\tilde{\mu})$, and it follows that $A(\mu)\cap A(\mu')\subseteq A(\tilde{\mu})$. Now suppose $a'\in A(\tilde{\mu})$. Then in fact $(u_R(a)-u_R(a'))\cdot\tilde{\mu}=0$, so that, by $\alpha\in(0,1)$, also both inequalities $(u_R(a)-u_R(a'))\cdot\mu\geq 0$ and $(u_R(a)-u_R(a'))\cdot\mu'\geq 0$ must hold with equality. Thus, $a'\in A(\mu)\cap A(\mu')$, and it follows that $A(\tilde{\mu})\subseteq A(\mu)\cap A(\mu')$.
\end{proof}

\subsection{The set $\cD$}\label{app:D}

A crucial aspect for an equilibrium with two posterior beliefs $\mu,\mu'\in\Delta\Theta$ is the ``size'' of the set $\cD=\set{r-r'}{r,r'\in\Delta A, \supp{r}\subseteq A(\mu), \textnormal{ and } \supp{r'}\subseteq A(\mu')}$. 

\begin{lemma}\label{lem:dimD}
The linear subspace spanned by $\cD$ has dimension $\abs{A(\mu)\cup A(\mu')}-1$.
\end{lemma}

\begin{proof}
%
We can identify $\Delta A$ and $\cD$ with subsets of $\R^{\abs{A}}.$ The restrictions that $r(a)-r'(a)=0$ for $a\notin A(\mu)\cup A(\mu')$ and $\vsum(r-r')=0$ show that the dimension of the span of $\cD$ is at most $\abs{A(\mu)\cup A(\mu')}-1$. It remains to find as many linearly independent vectors in $\cD$. If this number is zero, there is nothing to show. Otherwise, $A(\mu)$ and $A(\mu)'$ are nonempty sets whose union contains at least two elements. We can then partition their union into two nonempty disjoint sets $E$ and $E'$, with $E\subseteq A(\mu)$ and $E'\subseteq A(\mu')$, respectively. For any $a\in A$, let $\delta_a\in\Delta A$ be the Dirac measure that assigns probability one to $\{a\}$. Fix some $a\in E$ and $a'\in E'$. Then $\cD$ contains $\abs{E'}$ vectors of the form $\delta_a-\delta_{a''}$ with $a''\in E'$, and $\abs{E}-1$ vectors of the form $\delta_{a''}-\delta_{a'}$ with $a''\in E\setminus \{a\}$, and all these vectors are linearly independent.
\end{proof}

\subsection{Hyperplanes and half spaces}\label{app:hyperplanes}

Recall that, for any $x\in\R^n$, $H_{x}$, $H_{x}^{-}$, and $H_{x}^{+}$ respectively denote the sets $\set{y\in\R^n}{y\cdot x=0}$, $\set{y\in\R^n}{y\cdot x\leq 0}$, and $\set{y\in\R^n}{y\cdot x\geq 0}$. If $x\neq 0$, thus, $H_{x}$ is the hyperplane with normal vector $x$, and $H_{x}^{-}$ and $H_{x}^{+}$ are the two associated (closed) half spaces (whereas $H_{0}=H_{0}^{-}=H_{0}^{+}=\R^n$).

\begin{lemma}\label{lem:halfspaceshyperplanes}
Let $v_1,v_2\in\R^n$ and $D=\co\{v_1,v_2\}$. Then
\begin{enumerate}
\item\label{intersectionhalfspaces}
$\bigcap_{x\in D}H_{x}^{-}=H_{v_1}^{-}\cap H_{v_2}^{-}$,
\item\label{unionhyperplanes}
$\bigcup_{x\in D}H_{x}=(H_{v_1}^{-}\cap H_{v_2}^{+})\cup(H_{v_1}^{+}\cap H_{v_2}^{-})$, and
\item\label{interior}
$H_{v_1}^{-}\cap H_{v_2}^{-}$ is convex, and it has nonempty interior unless $v_2=\lambda v_1\neq 0$ for some $\lambda<0$.
\end{enumerate}
\end{lemma}

\begin{proof}
\ref{intersectionhalfspaces} Clearly, the set on the left is a subset of the set on the right. For the other direction, note that $y\cdot v_1\leq 0$ and $y\cdot v_2\leq 0$ implies $y\cdot(\alpha v_1+(1-\alpha)v_2)=\alpha y\cdot v_1+ (1-\alpha) y\cdot v_2\leq 0+0$ for $\alpha\in (0,1)$.

\ref{unionhyperplanes} Being element of the set on the right means having a nonnegative dot product with one of the values $v_1,v_2$ and a nonpositive one with the other. Between them, at an element $x$ of $D$, the value must be zero, which means being in some $H_{x}$. On the other hand, if $y\in H_{x}$ for some $x\in D$, we clearly cannot have both $y\cdot v_1<0$ and $y\cdot v_2<0$, or both $y\cdot v_1>0$ and $y\cdot v_2>0$. But this means $y$ must be in the set on the right.

\ref{interior} Convexity is straightforward. Concerning the nonempty interior, consider first the case that $v_1$ or $v_2$ is the null vector, say w.l.o.g.\ $v_2$. Then $H_{v_2}^{-}=\R^n$, so the intersection equals $H_{v_1}^{-}$. We always have $-v_1\in\inte H_{v_1}^{-}$, by $H_{v_1}^{-}=\R^n$ if also $v_1=0$, and otherwise by $(-v_1)\cdot v_1=-v_1^2<0$ and continuity of the dot product. Hence, suppose neither $v_1$ nor $v_2$ is the null vector. If they are linearly dependent, then $v_2=\lambda v_1$ for some $\lambda\neq 0$. So suppose now this is the case with $\lambda>0$. Then $H_{v_2}^{-}=H_{v_1}^{-}$, which has nonempty interior as already shown. Finally, consider the case that $v_1$ and $v_2$ are linearly independent. Then we can write $v_2=\lambda v_1+w$ with $v_1\cdot w=0$ and $w\neq 0$. Since also $v_1\neq 0$, this implies both $-(v_1+w)\cdot v_1=-v_1^2<0$ and $-(v_1+w)\cdot v_2=-\lambda v_1^2-w^2<0$. Thus, again by continuity of the dot product, $-(v_1+w)$ is in the interior of $H_{v_1}^{-}\cap H_{v_2}^{-}$.
\end{proof}

For later reference, the following lemma collects some standard facts about orthogonal complements; see \citet[Chapter 2]{Lax}.

\begin{lemma}\label{lem:intersectionofhyperplanes}
Let $D$ be a linear subspace of $\R^n$ and $U$ the orthogonal complement of $D$, that is the space of all $y\in\R^n$ such that $x\cdot y=0$ for all $x\in D$. Then $D$ is also the orthogonal complement of $U$ and $\dim D+\dim U=n$. Moreover, $U$ is the orthogonal complement of any set of vectors that span $D$. Similarly, $D$ is the orthogonal complement of any set of vectors that span $U$.
\end{lemma}

\section{Proofs of main results}\label{app:proofs}

\subsection{Proofs of Lemma~\ref{lem:samebeliefs} and Lemma~\ref{lem:caratheodory}}\label{app:lemmas}

\begin{proof}[{\bfseries Proof of Lemma~\ref{lem:samebeliefs}}]
The result follows from the following claim.

\begin{claim}\label{clm:samebeliefs}
Let $\Gamma$ be a cheap-talk game and $p$ a given Bayes plausible posterior belief distribution with finite support. Then there exists a neighborhood $V$ of $p$ such that, if $(\tilde{p},\tilde{\rho})$ is an equilibrium with $\tilde{p}\in V$, there is a sender payoff-equivalent equilibrium with posterior belief distribution $p$. Moreover, the neighborhood $V$ does not depend on the sender's utility function $u_S$.
\end{claim}

For each $\mu\in\supp{p}$, there exists a neighborhood $U_{\mu}$ such that $A(\mu')\subseteq A(\mu)$ for all $\mu'\in U_{\mu}$ (see Lemma~\ref{lem:Aneighborhood}). Further, since $\Theta$ is finite and the support correspondence is lower hemicontinuous by Theorem~17.14 in \citet{AliprantisBorder06}, there exists for each $\mu\in\supp{p}$ a neighborhood $V_{\mu}$ such that $\supp{\mu'}\supseteq \supp{\mu}$ for each $\mu'\in V_{\mu}$.\footnote
{In detail: For each $\theta\in\supp{\mu}$, there must exist a neighborhood $V_{\mu}^\theta$ of $\mu$ such that $\supp{\mu'}\cap\{\theta\}\neq\emptyset$ for $\mu'\in V_{\mu}^\theta$. Let $V_{\mu}$ be the intersection of these $V_{\mu}^\theta$.}
Again, since the support correspondence is lower hemicontinuous, there exists a neighborhood $V$ of $p$ such that for every $p'\in V$, and every
$\mu\in\supp{p}$, one has $\supp{p'}\cap U_{\mu}\cap V_{\mu}\neq\emptyset$. So far, nothing depended on $u_S$.

Now suppose $(\tilde{p},\tilde{\rho})$ is an equilibrium with $\tilde{p}\in V$. Since $\supp{\tilde{p}}\cap U_{\mu}\cap V_{\mu}\neq\emptyset$ for all $\mu\in\supp{p}$, there is a function $f_{\tilde{p}}:\supp{p}\to \supp{\tilde{p}}$ such that $f_{\tilde{p}}(\mu)\in U_{\mu}\cap V_{\mu}$ for all $\mu\in\supp{p}$. Let $\rho=\tilde{\rho}\circ f_{\tilde{p}}$. We claim that $(p,\rho)$ is an equilibrium that is sender payoff-equivalent to $(\tilde{p},\tilde{\rho})$. By assumption, $p$ is Bayes plausible. Moreover, since $A(\mu')\subseteq A(\mu)$ for $\mu'\in U_{\mu}$ and $f_{\tilde{p}}(\mu)\in U_{\mu}$, we have $\supp{(\tilde{\rho}\circ f_{\tilde{p}}(\mu))}\subseteq A(\mu)$ for all $\mu\in\supp{p}$. Now take any $\mu^*\in\supp{p}$ and $\theta\in\supp{\mu^*}$. Since $(\tilde{p},\tilde{\rho})$ is an equilibrium, it must be optimal to choose any $\mu'\in\supp{\tilde{p}}$ that has $\theta$ in its support at $\theta$. This gives us the second equality in the following:
\begin{align*}
u_S(\theta)\cdot\rho(\mu^*)=u_S(\theta)\cdot\tilde{\rho}(f_{\tilde{p}}(\mu^*))&=\max_{\mu'\in\supp{\tilde{p}}}u_S(\theta)\cdot\tilde{\rho}(\mu') \\
&\geq\max_{\mu\in\supp{p}}u_S(\theta)\cdot\tilde{\rho}(f_{\tilde{p}}(\mu)) \\
&=\max_{\mu\in\supp{p}}u_S(\theta)\cdot\rho(\mu) \\
&\geq u_S(\theta)\cdot\rho(\mu^*).
\end{align*}
Therefore, all inequalities are actually equalities. The equality
\begin{equation*}
u_S(\theta)\cdot\rho(\mu^*)=\max_{\mu\in\supp{p}}u_S(\theta)\cdot\rho(\mu)
\end{equation*}
means that the last condition in the definition of an equilibrium is satisfied for $(p,\rho)$ and the equality
\begin{equation*}
u_S(\theta)\cdot\rho(\mu^*)=\max_{\mu'\in\supp{\tilde{p}}}u_S(\theta)\cdot\tilde{\rho}(\mu')
\end{equation*}
together with $(\tilde{p},\tilde{\rho})$ being an equilibrium guarantees that $(p,\rho)$ is sender payoff-equivalent.
\end{proof}

\begin{proof}[{\bfseries Proof of Lemma~\ref{lem:caratheodory}}]
By equilibrium condition~\ref{eqm_p}, $p$ must be Bayes plausible, i.e., $\mu_0$ must be a convex combination of all beliefs $\mu \in \supp{p}$. By Carath{\'e}odory's theorem, there is, thus, a subset $B$ of $\supp{p}$ with $|B| \leq |\Theta|$ such that $\mu_0$ is a convex combination of all beliefs $\mu \in B$ (since one can identify beliefs $\mu \in \Delta \Theta$ with vectors in $\R^{|\Theta|-1}$ by a linear injective mapping). This convex combination defines a Bayes plausible belief distribution $p'$ with $\supp{p'}\subseteq B\subseteq\supp{p}$. Now let $\rho'=\rho\vert_{\supp{p'}}$. Then $\rho'=\rho\circ f$, where $f$ is the identity function on $\supp{p'}$. This function $f$ trivially satisfies the properties that were used in the proof of Lemma~\ref{lem:samebeliefs}, so the arguments given there (applied to $p'$ in place of $p$ and $(p,\rho)$ in place of $(\tilde{p},\tilde{\rho})$) show that $(p',\rho')$ is an equilibrium that is sender payoff-equivalent to $(p,\rho)$.

Next, suppose $(p,\rho)$ is robust and let, for any neighborhoods $U$, $V$, and $W$, $U'$ be the corresponding set of perturbed sender utility functions $\tilde{u}_S$. Then the same set $U'$ and the same construction that was used for $\rho'$ can also be applied for showing that $(p',\rho')$ is robust. Indeed, consider any $\tilde{u}_S\in U'$ and the corresponding equilibrium $(\tilde{p},\tilde{\rho})$. By Lemma~\ref{lem:samebeliefs}, we may assume $\tilde{p}=p$. Now let $\tilde{\rho}'=\tilde{\rho}\vert_{\supp{p'}}$. Then, by the arguments given for $(p',\rho')$, also $(p',\tilde{\rho}')$ is an equilibrium, and $(p',\tilde{\rho}')$ and $(p,\tilde{\rho})$ are sender payoff-equivalent. Since also $(p',\rho')$ and $(p,\rho)$ are sender payoff-equivalent, it follows that $(p',\rho')$ inherits robustness from $(p,\rho)$. The argument for a fully robust equilibrium is completely analogous.
\end{proof}

\subsection{Proof of Proposition~\ref{prop:binarybeliefrobust}}\label{app:binarybeliefrobust}

\noindent
{\bfseries\itshape Preliminaries.}
Recall that, for any $x\in\R^n$, $H_{x}$, $H_{x}^{-}$, and $H_{x}^{+}$ respectively denote the sets $\set{y\in\R^n}{y\cdot x=0}$, $\set{y\in\R^n}{y\cdot x\leq 0}$, and $\set{y\in\R^n}{y\cdot x\geq 0}$.

To start the proof, suppose $(p,\rho)$ is an equilibrium with $\supp{p}=\{\mu,\mu'\}$ and where $\mu<\mu'$. By Lemma~\ref{lem:samebeliefs}, we only need to consider equilibria $(\tilde{p},\tilde{\rho})$ for any perturbed utility function $\tilde{u}_S$ such that $\tilde{p}=p$. Thus, we only consider receiver strategies $\tilde{\rho}$ with domain $\{\mu,\mu'\}$, and it will be convenient to map any such $\tilde{\rho}$ to the difference $x=\tilde{\rho}(\mu)-\tilde{\rho}(\mu')$.

The set of these differences such that $(p,\tilde{\rho})$ satisfies equilibrium condition~\ref{eqm_r} is
\begin{equation*}
\cD:=\setd{r-r'}{\text{$r,r'\in\Delta A$, $\supp{r}\subseteq A(\mu)$, and $\supp{r'}\subseteq A(\mu')$}}.
\end{equation*}
Further, given $x=\tilde{\rho}(\mu)-\tilde{\rho}(\mu')$, and since $\mu<\mu'$ implies that $\theta_1\in\supp{\mu}$ and $\theta_2\in\supp{\mu'}$, equilibrium condition~\ref{eqm_s} is satisfied if and only if
\begin{equation*}
\tilde{u}_S(\theta_1)\cdot x\geq 0\quad\text{and}\quad \tilde{u}_S(\theta_2)\cdot x\leq 0
\end{equation*}
for $x=\tilde{\rho}(\mu)-\tilde{\rho}(\mu')$, where the first inequality must hold with equality if $\mu'<1$ (because then also $\theta_1\in\supp{\mu'}$), and the second must be an equality if $\mu>0$ (because then also $\theta_2\in\supp{\mu}$).

Thus, there exists an equilibrium $(p,\tilde{\rho})$ for $\tilde{u}_S$ if and only if there is some $x\in\cD$ that satisfies the two given (in)equalities. In particular, since $(p,\rho)$ is an equilibrium for $\tilde{u}_S=u_S$, $x=\rho(\mu)-\rho(\mu')$ satisfies $x\in\cD$ and in fact $u_S(\theta_1)\cdot x=u_S(\theta_2)\cdot x=0$, because transparency of $\Gamma$ means $u_S(\theta_1)=u_S(\theta_2)$.

\noindent
{\bfseries\itshape Proof of necessity.}
Suppose now that $(p,\rho)$ is robust and generates an expected sender payoff $s>\max\cV(\mu_0)$. Note that $\mu<\mu'$ implies in fact $\mu<\mu_0<\mu'$. By way of contradiction, suppose none of the conditions \ref{fullyinf}--\ref{fouractions} holds. Specifically, assume $\mu'<1$, as the case in which $\mu'=1$ and $\mu>0$ is symmetric to $\mu=0$ and $\mu'<1$.

Fix any neighborhoods $U$, $V$, and $W$ from the definition of a robust equilibrium and let $U'$ be the associated set of perturbed sender utility functions. Take any $\tilde{u}_S\in U'$ and let $(\tilde{p},\tilde{\rho})$ be the corresponding equilibrium, where we may assume $\tilde{p}=p$ by Lemma~\ref{lem:samebeliefs} and hence consider $x=\tilde{\rho}(\mu)-\tilde{\rho}(\mu')$. As argued at the beginning of the proof, necessarily $x\in\cD$, and, since $\mu'<1$, $\tilde{u}_S(\theta_1)\cdot x=0$ and $\tilde{u}_S(\theta_2)\cdot x\leq 0$, where the inequality cannot be strict if $\mu>0$. We are now going to argue that this implies $x=0$ (if $\tilde{u}_S$ is a generic element of $U'$).

By Lemma~\ref{lem:dimD}, $\cD$ is contained in a subspace with dimension $m=\abs{A(\mu)\cup A(\mu')}-1<n$. Therefore, by Lemma~\ref{lem:intersectionofhyperplanes}, $\cD\subseteq\bigcap_{i=1}^{n-m}H_{v_i}$ for some linearly independent vectors $v_1,\dots,v_{n-m}$, which we fix. Since condition~\ref{fouractions} does not hold, we have $m<3$, so there are three cases to consider. If $m=0$, then $\cD=\{0\}$, which trivially implies $x=0$. If $m=1$, then necessarily $x\in H_{\tilde{u}_S(\theta_1)}\cap\bigcap_{i=1}^{n-1}H_{v_i}$, where, generically for all $\tilde{u}_S\in\tilde{\cU}$, the vectors $\tilde{u}_S(\theta_1),v_1,\dots,v_{n-1}$ are linearly independent, which implies $x=0$. If $m=2$, i.e., $\abs{A(\mu)\cup A(\mu')}=3$, then, since condition~\ref{threeactions} does not hold, we must have $\mu>0$, so necessarily $x\in H_{\tilde{u}_S(\theta_1)}\cap H_{\tilde{u}_S(\theta_2)}\cap\bigcap_{i=1}^{n-2}H_{v_i}$, where, again generically, the vectors $\tilde{u}_S(\theta_1),\tilde{u}_S(\theta_2),v_1,\dots,v_{n-2}$ are linearly independent, which once more requires $x=0$.

However, $x=0$ means $\tilde{\rho}(\mu)=\tilde{\rho}(\mu')$, and then the sender's interim expected equilibrium payoff in any state is $\tilde{s}^*(\theta)=\tilde{u}_S(\theta)\cdot\tilde{\rho}(\mu)$. Further, $x=0\in\cD$ only if $\supp{(\tilde{\rho}(\mu))}=\supp{(\tilde{\rho}(\mu'))}\subseteq A(\mu)\cap A(\mu')$. Then, by Lemma~\ref{lem:Aintersection} and $\mu<\mu_0<\mu'$, also $\supp{(\tilde{\rho}(\mu))}\subseteq A(\mu_0)$. The latter means that there is also a babbling equilibrium in which the receiver uses the mixed action $\tilde{\rho}(\mu)$ (for the ``posterior'' belief $\mu_0$). Using the state-independent utility function $u_S$, Lemma~LR implies $u_S(\theta)\cdot\tilde{\rho}(\mu)\in\cV(\mu_0)$ for any $\theta$. Thus, since $\tilde{s}^*(\theta)=\tilde{u}_S(\theta)\cdot\tilde{\rho}(\mu)$, $\tilde{s}^*(\theta)$ must be in the vicinity of $\cV(\mu_0)$ if we take $U$ to be small enough and $U'$ a subset of $U$ (which is w.l.o.g.). This, however, contradicts that $\tilde{s}^*(\theta)$ approaches $s>\max\cV(\mu_0)$ if we take $W$ to be small enough.

\noindent
{\bfseries\itshape Proof of sufficiency.}
For the first step, we are going to use the facts established at the beginning of the proof to show that, under any of the conditions \ref{fullyinf}--\ref{fouractions}, and for any neighborhoods $U$ of $u_S$ and $X$ of $\rho(\mu)-\rho(\mu')$, there is some $u_S'\in U$ with a neighborhood $U'\subseteq U$ such that, for any $\tilde{u}_S\in U'$, there is an equilibrium $(p,\tilde{\rho})$ such that $x=\tilde{\rho}(\mu)-\tilde{\rho}(\mu')\in X$. The candidate $x=\rho(\mu)-\rho(\mu')$, which is in $\cD\cap X$, trivially works for all $\tilde{u}_S$ if $x=0$. Hence, suppose $\rho(\mu)\neq\rho(\mu')$.

If condition~\ref{fullyinf} holds, i.e., $\mu=0$ and $\mu'=1$, then $x=\rho(\mu)-\rho(\mu')\neq 0$ still works for a perturbed utility function $\tilde{u}_S$ if and only if $\tilde{u}_S(\theta_1)\cdot x\geq 0$ and $\tilde{u}_S(\theta_2)\cdot x\leq 0$. This condition can be written as $\tilde{u}_S\in H_{x}^{+}\times H_{x}^{-}$. Being half spaces, both $H_{x}^{+}$ and $H_{x}^{-}$ are convex and have nonempty interior, and then the same holds for their product. Moreover, since $u_S(\theta_1)\cdot x=u_S(\theta_2)\cdot x=0$, in particular $u_S\in H_{x}^{+}\times H_{x}^{-}$. Therefore, by Lemma 5.28 (1.) in \citet{AliprantisBorder06}, the set $U\cap H_{x}^{+}\times H_{x}^{-}$ has nonempty interior. Thus, we can take this interior as $U'$ if condition~\ref{fullyinf} holds.

Now suppose condition~\ref{threeactions} or condition~\ref{fouractions} holds, but condition~\ref{fullyinf} not. Specifically, assume again $\mu'<1$, as the case in which $\mu'=1$ and $\mu>0$ is symmetric to $\mu=0$ and $\mu'<1$. Then the equilibrium (in)equalities for $x$ are $\tilde{u}_S(\theta_1)\cdot x=0$ and $\tilde{u}_S(\theta_2)\cdot x\leq 0$, where the inequality cannot be strict if $\mu>0$. Now we cannot use $x=\rho(\mu)-\rho(\mu')$ for many perturbed utility functions $\tilde{u}_S$ anymore, and the more involved arguments that establish the desired set $U'$ are relegated to Proposition~\ref{prop:opensetoneindiffonepref} (in case $\mu=0$) and Proposition~\ref{prop:opensettwoindiff} (in case $\mu>0$). To see that the prerequisites for these propositions are satisfied, recall that $x=\rho(\mu)-\rho(\mu')$ satisfies $x\in\cD$ and $u_S(\theta_1)\cdot x=u_S(\theta_2)\cdot x=0$. Moreover, Lemma~\ref{lem:dimD} implies that $\cD$ contains $m=\abs{A(\mu)\cup A(\mu')}-1$ linearly independent vectors, which, since condition~\ref{threeactions} or condition~\ref{fouractions} holds, is at least two, and at least three if $\mu>0$. Further, if we pick any number of linearly independent vectors from $\cD$, we may, by the Steinitz exchange lemma, assume that $x=\rho(\mu)-\rho(\mu')$ is one of them, because then $x\neq 0$ by hypothesis. Therefore, as $\cD$ is convex, we can indeed apply Proposition~\ref{prop:opensetoneindiffonepref} or Proposition~\ref{prop:opensettwoindiff} to obtain $U'$.

For the second step of proving sufficiency, note that, if $U'$ is the subset of $U$ established in the first step, then we can ensure that all $\tilde{u}_S\in U'$ are arbitrarily close to $u_S$ by choosing $U$ small enough. Additionally, by choosing $X$ small enough and using Proposition~\ref{prop:decomp}, we can guarantee that $\tilde{\rho}(\mu)$ and $\tilde{\rho}(\mu')$ are, respectively, arbitrarily close to $\rho(\mu)$ and $\rho(\mu')$ (uniformly for all $\tilde{u}_S\in U'$). This way, we can let the sender's equilibrium payoffs $\tilde{s}^*(\theta)=\max(\tilde{u}_S(\theta)\cdot\tilde{\rho}(\mu),\tilde{u}_S(\theta)\cdot\tilde{\rho}(\mu'))$ be as close to $s^*(\theta)=\max(u_S(\theta)\cdot\tilde{\rho}(\mu),u_S(\theta)\cdot\tilde{\rho}(\mu'))$ as we want for every $\theta$ (again uniformly for all $\tilde{u}_S\in U'$), which yields the robustness of $(p,\rho)$.

The proof of Proposition~\ref{prop:binarybeliefrobust} is now complete up to the three more technical Propositions \ref{prop:opensetoneindiffonepref}, \ref{prop:opensettwoindiff}, and \ref{prop:decomp}, which are central for sufficiency.

\begin{proposition}\label{prop:opensetoneindiffonepref}
Let $u_0\in\R^n$, $D=\co\{v_1,v_2\}$ for two linearly independent vectors $v_1,v_2\in\R^n$, and $x_0\in D$ such that $u_0\cdot x_0=0$. Then, for any given $\epsilon>0$, there exists a nonempty, open set $U\subseteq\R^n\times\R^n$ such that $U\subseteq B_{\epsilon}(u_0,u_0)$ and for all $(u_1,u_2)\in U$ there is some $x\in D\cap B_{\epsilon}(x_0)$ such that $u_1\cdot x=0\geq u_2\cdot x$ (i.e., $U\subseteq\bigcup_{x\in D\cap B_{\epsilon}(x_0)}H_{x}\times H_{x}^{-}$).
\end{proposition}

\begin{proof}
Since $x_0\in D$, there is some $\lambda_0\in\R_+^2\setminus\{0\}$ such that, in matrix notation, $(v_1,v_2)\lambda_0=x_0$. In particular, thus, $x_0\neq 0$ by linear independence of $v_1,v_2$. First suppose $\lambda_0\in\inte\R_+^2=\R_{++}^2$.

The proof strategy is to identify a convex subset $D'$ of $D\cap B_{\epsilon}(x_0)$ and two nonempty open sets $U_1,U_2\subseteq B_{\epsilon/2}(u_0)$ such that for every $u_1\in U_1$ we have $u_1\cdot x=0$ for \emph{some} $x\in D'$, and for every $u_2\in U_2$ we have $u_2\cdot x\leq 0$ for \emph{all} $x\in D'$. Then, for every pair $(u_1,u_2)\in U_1\times U_2$, by $u_1\in U_1$ there is some $x\in D'$ such that $u_1\cdot x=0$, and in particular for this $x$ also $u_2\cdot x\leq 0$ by $u_2\in U_2$, so together indeed $u_1\cdot x=0\geq u_2\cdot x$. (In the set notation, we exploit the fact that $\bigcup_{x\in D'}H_{x}\times\bigcap_{x\in D'}H_{x}^{-}\subseteq\bigcup_{x\in D'}H_{x}\times H_{x}^{-}$ for any subset $D'$ of $D$.) Hence, we can choose $U=U_1\times U_2$, because it is a subset of $B_{\epsilon}(u_0,u_0)$ due to $U_1,U_2\subseteq B_{\epsilon/2}(u_0)$.

We are now going to construct a set $D'$ such that $u_0\cdot x\leq 0$ for all $x\in D'$, and then we will verify the actually needed properties. Therefore, note that $\lambda_0\in\R_{++}^2$ and $u_0\cdot x_0=(u_0\cdot v_1,u_0\cdot v_2)\lambda_0=0$ together imply that either $u_0\cdot v_1<0<u_0\cdot v_2$, or $u_0\cdot v_2<0<u_0\cdot v_1$, or $u_0\cdot v_1=u_0\cdot v_2=0$. In the last case, in fact $u_0\cdot x\leq 0$ for all $x\in D$, so let $D'=\co\{x_0,x'\}$ for an arbitrary other vector $x'\neq x_0$ from $D\cap B_{\epsilon}(x_0)$. In the first case, $u_0\cdot x\leq 0$ for all $x\in\co\{v_1,x_0\}$, so let $D'=\co\{x_0,x'\}$ for an arbitrary $x'\neq x_0$ from $\co\{v_1,x_0\}\cap B_{\epsilon}(x_0)$, noting that such an $x'$ exists due to $x_0\neq v_1$. In the second case let analogously $D'=\co\{x_0,x'\}$ for an arbitrary $x'\neq x_0$ from $\co\{v_2,x_0\}\cap B_{\epsilon}(x_0)$. In any case we now have a set $D'\subseteq D\cap B_{\epsilon}(x_0)$ such that $u_0\cdot x\leq 0$ for all $x\in D'$, which we can write as $u_0\in\bigcap_{x\in D'}H_{x}^{-}$. Further, in any case $D'=\co\{x_0,x'\}$ for some $x'\in D$ that is distinct from $x_0$, so that these two vectors inherit linear independence from $v_1$ and $v_2$.

Consider the property $u_0\in\bigcap_{x\in D'}H_{x}^{-}$. By Lemma~\ref{lem:halfspaceshyperplanes}, the latter intersection is convex and has nonempty interior. Thus, also the open set $B_{\epsilon/2}(u_0)\cap\inte\bigcap_{x\in D'}H_{x}^{-}$ is nonempty by Lemma 5.28 (1.) in \citet{AliprantisBorder06}, so we can use it as our set $U_2$, because for every $u_2\in U_2$ then $u_2\cdot x\leq 0$ holds by construction for all $x\in D'$ (i.e., $U_2\subseteq\bigcap_{x\in D'}H_{x}^{-}$).

By $x_0\in D'$ and $u_0\cdot x_0=0$, we further have $u_0\in\bigcup_{x\in D'}H_{x}$. By Lemma~\ref{lem:halfspaceshyperplanes}, and since $D'=\co\{x_0,x'\}$ for two linearly independent vectors $x_0$ and $x'$, $\bigcup_{x\in D'}H_{x}=(H_{x_0}^{-}\cap H_{x'}^{+})\cup(H_{x_0}^{+}\cap H_{x'}^{-})$. Each of the two latter intersections, again by Lemma~\ref{lem:halfspaceshyperplanes}, is convex and has nonempty interior. Choose one of them so that it contains $u_0$. Then the interior of this convex set has a nonempty intersection with $B_{\epsilon/2}(u_0)$ by Lemma 5.28 (1.) in \citet{AliprantisBorder06}. Therefore, we can use the latter intersection as our set $U_1$, because for every $u_1\in U_1$ then $u_1\cdot x=0$ by construction for some $x\in D'$ (i.e., $U_1\subseteq\bigcup_{x\in D'}H_{x}$). Letting $U=U_1\times U_2$ as argued in the beginning completes the proof for $\lambda_0\in\R_{++}^2$.

Now suppose $\lambda_0\not\in\R_{++}^2$. Consider any $\delta>0$ and let $\hat\lambda=\lambda_0+\delta\one$, which is in $\R_{++}^2$ (in particular $\hat\lambda\neq 0$). Then let $\hat x=\frac{1}{\one\cdot\hat\lambda}(v_1,v_2)\hat\lambda$, which is in $D$. We can induce $\hat x$ to be arbitrarily close to $x_0$ by starting with sufficiently small $\delta$, because $(v_1,v_2)\lambda_0=x_0$ and $\one\cdot\lambda_0=1$. Next, let $\hat u$ be the orthogonal projection of $u_0$ onto the hyperplane $H_{\hat x}$ (so that $\hat u\cdot\hat x=0$), noting that $\hat x\in D$ implies $\hat x\neq 0$. To make $\norm{\hat u-u_0}$ arbitrarily small, we only need to make $\norm{\hat x-x_0}$ small enough---which we can do through $\delta$ (cf.\ footnote~\ref{fn:projection}). Let $\delta$ in fact be sufficiently small so that $\hat u\in B_{\epsilon/2}(u_0)$ and $\hat x\in B_{\epsilon}(x_0)$. Then there exists also a sufficiently small $\epsilon'>0$ so that $B_{\epsilon'}(\hat u)\subseteq B_{\epsilon/2}(u_0)$ and $B_{\epsilon'}(\hat x)\subseteq B_{\epsilon}(x_0)$. Fix such an $\epsilon'$. Since $\hat\lambda\in\R_{++}^2$, we can apply the already proved results for $\lambda_0\in\R_{++}^2$ to $\hat u$, $\hat x$, and $\epsilon'$ in place of $u$, $x_0$, and $\epsilon$, respectively. Thus, there exists a nonempty, open set $U\subseteq\R^n\times\R^n$ such that $U\subseteq B_{\epsilon'}(\hat u,\hat u)$ and $U\subseteq\bigcup_{x\in D\cap B_{\epsilon'}(\hat x)}H_{x}\times H_{x}^{-}$. We can actually use the same set $U$ for $u$, $x_0$, and $\epsilon$, because $B_{\epsilon'}(\hat x)\subseteq B_{\epsilon}(x_0)$, and also $B_{\epsilon'}(\hat u,\hat u)\subseteq B_{\epsilon}(u_0,u_0)$ by $B_{\epsilon'}(\hat u)\subseteq B_{\epsilon/2}(u_0)$.
\end{proof}


\begin{proposition}\label{prop:opensettwoindiff}
Let $u_0\in\R^n$, $D=\co\{v_1,v_2,v_3\}$ for three linearly independent vectors $v_1,v_2,v_3\in\R^n$, and $x_0\in D$ such that $u_0\cdot x_0=0$. Then, for any given $\epsilon>0$, there exists a nonempty, open set $U\subseteq\R^n\times\R^n$ such that $U\subseteq B_{\epsilon}(u_0,u_0)$ and for all $(u_1,u_2)\in U$ there is some $x\in D\cap B_{\epsilon}(x_0)$ such that $u_1\cdot x=u_2\cdot x=0$ (i.e., $U\subseteq\bigcup_{x\in D\cap B_{\epsilon}(x_0)}H_{x}\times H_{x}$).
\end{proposition}

\begin{proof}
Note that the condition $u_1\cdot x=u_2\cdot x=0$ holds for some $x\in D$ if and only if it holds for some $x=\sum_{i=1}^3\lambda_iv_i$ such that $\lambda=(\lambda_1,\lambda_2,\lambda_3)^\top\in\R_+^3\setminus\{0\}$. Hence, if we first set aside the additional requirement that $x\in B_{\epsilon}(x_0)$, we are looking for pairs $(u_1,u_2)\in\R^n\times\R^n$ such that the homogeneous system of linear equations
\begin{equation}\label{kerUX}
\begin{pmatrix} u_1\cdot v_1 & u_1\cdot v_2 & u_1\cdot v_3 \\ u_2\cdot v_1 & u_2\cdot v_2 & u_2\cdot v_3 \end{pmatrix}\lambda=\begin{pmatrix} 0 \\ 0 \end{pmatrix},\quad\lambda\in\R^3,
\end{equation}
has a solution in $\R_+^3\setminus\{0\}$. The latter is the case if $u_1=u_2=u_0$, since $u_0\cdot x_0=0$ and $x_0\in\cD$; then a suitable solution is the unique $\lambda_0\in\R^3$ such that, in matrix notation, $(v_1,v_2,v_3)\lambda_0=x_0$. Because we ultimately want to use the implicit function theorem to obtain the open set $U$, we are going to construct another, but near, starting pair $(\hat u_1,\hat u_2)$, which is such that the two rows of the matrix in \eqref{kerUX} are linearly independent and there is a solution $\hat\lambda$ in the \emph{interior} of the positive orthant. We indeed have $\hat\lambda\in\inte\R_+^3=\R_{++}^3$ if $\hat\lambda=\lambda_0+\delta\one$ for some $\delta>0$. Fix an arbitrary such $\hat\lambda$ (so in particular $\hat\lambda\neq 0$), let $y_0=(u_0\cdot v_1,u_0\cdot v_2,u_0\cdot v_3)^\top$, which lies in $H_{\lambda_0}\subseteq\R^3$, and let $\hat y$ be the projection of $y_0$ onto the hyperplane $H_{\hat\lambda}$. To have $\norm{\hat y-y_0}$ arbitrarily small, we only need $\norm{\hat\lambda-\lambda_0}$ to be small enough---which we can achieve by starting with sufficiently small $\delta$.\footnote
{Specifically, $\hat y=y_0+\mu\hat\lambda$ for $\mu=-y_0\cdot\hat\lambda/\hat\lambda^2$, so $\norm{\hat y-y_0}^2=(y_0\cdot\hat\lambda)^2/\hat\lambda^2$. The latter vanishes as $\hat\lambda\to\lambda_0$, because the numerator tends to $(y_0\cdot\lambda_0)^2=0$ and the denominator to $\lambda_0^2$, which is strictly positive since $\lambda_0\neq 0$.\label{fn:projection}}
Next, consider any two linearly independent vectors $b_1,b_2\in H_{\hat\lambda}$ and any $\delta'>0$, and let $y_i=\hat y+\delta'b_i$ for both $i=1,2$. Then $y_1,y_2\in H_{\hat\lambda}$ by construction, and $y_1,y_2$ are also linearly independent whenever $\delta'$ is small enough.\footnote
{If $y_1$ and $y_2$ are linearly \emph{dependent}, then one of them, say w.l.o.g. $y_2$, is a scalar multiple of the other, i.e., $y_2=\mu y_1$ for some $\mu\in\R$. This requires that $(1-\mu)\hat y=\delta'(\mu b_1-b_2)$, which is a nontrivial linear combination of $b_1,b_2$ (due to $\delta'>0$) and thus not null, implying that also $1-\mu\neq 0$. Hence, $\hat y$ must be a linear combination of $b_1,b_2$, with some coefficients $\mu_1,\mu_2\in\R$ that are uniquely determined by linear independence, and which by the previous equation must satisfy $\mu_1=\delta'\mu/(1-\mu)$ and $\mu_2=-\delta'/(1-\mu)$. The two latter equations yield $\mu(\delta'+\mu_1)=\mu_1$ and $\mu\mu_2=\delta'+\mu_2$, so $(\delta'+\mu_1)(\delta'+\mu_2)=\mu_1\mu_2$, which holds for $\delta'=0$ and hence for at most one $\delta'>0$.}
By linear independence of $v_1,v_2,v_3$, the linear function $f$ that maps every $u\in\R^n$ to $(u\cdot v_1,u\cdot v_2,u\cdot v_3)^\top=(v_1,v_2,v_3)^\top u\in\R^3$ is surjective. Hence, there are two vectors $\hat u_1,\hat u_2\in\R^n$ such that $f(\hat u_i)=y_i$ for both $i=1,2$, so that the starting pair $(\hat u_1,\hat u_2)$ has the desired properties by construction. Moreover, we may assume it to be arbitrarily close to $(u_0,u_0)$ by choosing sufficiently small $\delta$ and $\delta'$. Indeed, $f$ is an open mapping by Theorem 5.18 in \citet{AliprantisBorder06}, and hence the inverse correspondence $f^{-1}$ is lower hemicontinuous by Theorem~17.7 (ibid.), so if $y_i$ is close enough to $y_0$ (which we can achieve via the triangle inequality by making $\delta$ and $\delta'$ small), there is some $\hat u_i\in f^{-1}(y_i)$ as close as desired to $u_0\in f^{-1}(y_0)$.

In summary, for any such pair $(\hat u_1,\hat u_2)$, the matrix in \eqref{kerUX} equals $(y_1,y_2)^\top$ for some linearly independent vectors $y_1,y_2\in\R^3$, and there is a corresponding solution $\hat\lambda\in\inte\R_+^3$, so that actually $y_1,y_2\in H_{\hat\lambda}$. Since $\one\cdot\hat\lambda>0$, the vector $\one\in\R^3$ is outside $H_{\hat\lambda}$, and thus the square matrix $(y_1,y_2,\one)^\top$ has full rank. Moreover, $\hat\lambda$ satisfies $(y_1,y_2,\one)^\top\hat\lambda=(0,0,\one\cdot\hat\lambda)^\top$ by construction.

Fixing $\hat u_1$, $\hat u_2$, and $\hat\lambda$, we are now in the position to apply the implicit function theorem. Let $F$ be the (continuously differentiable) mapping that assigns the value
\begin{equation*}
F(u_1,u_2,\lambda)=\begin{pmatrix} u_1\cdot v_1 & u_1\cdot v_2 & u_1\cdot v_3 \\ u_2\cdot v_1 & u_2\cdot v_2 & u_2\cdot v_3 \\ 1 & 1 & 1 \end{pmatrix}\lambda-\begin{pmatrix} 0 \\ 0 \\ \one\cdot\hat\lambda \end{pmatrix}\in\R^3
\end{equation*}
to every triple $(u_1,u_2,\lambda)\in\R^n\times\R^n\times\R^3$. Then $F(\hat u_1,\hat u_2,\hat\lambda)=0$, and the Jacobian of $F$ with respect to $\lambda$ and evaluated at $(u_1,u_2,\lambda)=(\hat u_1,\hat u_2,\hat\lambda)$ is $(y_1,y_2,\one)^\top$, which has full rank. Hence, there exists a neighborhood $\hat U\subseteq\R^n\times\R^n$ of $(\hat u_1,\hat u_2)$ and a continuously differentiable function $g\colon\hat U\to\R^3$ such that $g(\hat u_1,\hat u_2)=\hat\lambda$ and $F(u_1,u_2,g(u_1,u_2))=0$ for all $(u_1,u_2)\in\hat U$. Since $\hat\lambda\in\inte\R_+^3$ and $g$ is continuous, there is another neighborhood $U\subseteq\hat U$ of $(\hat u_1,\hat u_2)$ such that $g(U)\subseteq\inte\R_+^3$. Therefore, for every pair $(u_1,u_2)\in U$, there is some $\lambda=g(u_1,u_2)\in\inte\R_+^3$ such that $F(u_1,u_2,\lambda)=0$, so that in particular \eqref{kerUX} holds. Forcing $\hat u_1$ and $\hat u_2$ to be sufficiently close to $u_0$ (through small $\delta$ and $\delta'$), and keeping the radius of $U$ small enough (but still positive), we can further ensure that $U\subseteq B_{\epsilon}(u_0,u_0)$.

The last requirement to fulfill is that $x=(v_1,v_2,v_3)\lambda=(v_1,v_2,v_3)g(u_1,u_2)$ stays in $B_{\epsilon}(x_0)$ for all $(u_1,u_2)\in U$. By $x_0=(v_1,v_2,v_3)\lambda_0$, it is enough to keep $\lambda$ close to $\lambda_0$, or $\lambda$ close to $\hat\lambda$ and $\hat\lambda$ close to $\lambda_0$. We can indeed make $\norm{\hat\lambda-\lambda_0}=\norm{\delta\one}$ arbitrarily small through $\delta$, whereas $\norm{\lambda-\hat\lambda}=\norm{g(u_1,u_2)-g(\hat u_1,\hat u_2)}$ can be kept arbitrarily small for all $(u_1,u_2)\in U$ by a sufficiently small radius of $U$, because $g$ is continuous.
\end{proof}


To state Proposition~\ref{prop:decomp}, let $\mu,\mu'\in\Delta\Theta$ now be two arbitrary beliefs and consider the set
\begin{equation*}
\cR:=\setd{(r,r')\in(\Delta A)^2}{\text{$\supp{r}\subseteq A(\mu)$ and $\supp{r'}\subseteq A(\mu')$}}
\end{equation*}
and the (onto) mapping from $\cR$ to $\cD$ that maps $(r,r')$ to $x=r-r'$. We need an inverse that is continuous in an arbitrary given point $x_0=r_0-r_0'$.

\begin{proposition}\label{prop:decomp}
Let $x_0=r_0-r_0'\in\cD$. For any neighborhood $R$ of $(r_0,r_0')$, there is a neighborhood $X$ of $x_0$ such that every $x\in X\cap\cD$ has a representation $x=r-r'$ with $(r,r')\in R\cap\cR$.
\end{proposition}

The proof of Proposition~\ref{prop:decomp} will use the following lemma, which characterizes the whole preimage $\cR_x:=\set{(r,r')\in\cR}{r-r'=x}$ of $x\in\cD$. For any vector $x=(x_1,\dots,x_n)^\top\in\R^n$, let $x^+$ denote $(\max(x_1,0),\dots,\max(x_n,0))^\top\in\R_+^n$, and let $x^-$ denote $(-x)^+$ (so that $x=x^+-x^-$).

\begin{lemma}\label{lem:decomp}
Let $x\in\cD$. Then $(r,r')\in\cR_x$ if and only if $(r,r')=(x^++d,x^-+d)$ for some $d=(d_1,\dots,d_n)^\top\in\R_+^n$ such that $\vsum d=1-\vsum x^+$ and $d_i=0$ whenever $a_i\not\in A(\mu)\cap A(\mu')$.
\end{lemma}

\begin{proof}
``$\Rightarrow$'': Suppose $(r,r')\in\cR_x$ and let $d=r-x^+$. Then $r=x^++d$ and $r'=r-x=r-(x^+-x^-)=d+x^-$. Further, since $r\geq 0$ and $r=x+r'$, where also $r'\geq 0$, we have $r\geq x^+$, which implies $d\in\R_+^n$, and $\vsum d=1-\vsum x^+$ by $\vsum r=1$. Now let $d_i$ be the $i$th coordinate of $d$. If $a_i\not\in A(\mu)$, then $r(a_i)=0$, which implies $d_i=0$, because $r=x^++d$ and both $x^+$ and $d$ are nonnegative. Likewise, if $a_i\not\in A(\mu')$, then $r'(a_i)=0$, and hence $d_i=0$ by $r'=x^-+d$ and nonnegativity of $x^-$ and $d$. Thus, $d$ is as claimed.

``$\Leftarrow$'': Suppose $r=x^++d$ and $r'=x^-+d$ for some $d$ as stated in the lemma. Then both $r$ and $r'$ are in $\R_+^n$, and $r-r'=x$. Further, $\vsum r=\vsum x^++\vsum d=1$ and, since $x\in\cD$ implies $\vsum x=0$ and thus $\vsum x^-=\vsum x^+$, also $\vsum r'=\vsum x^-+\vsum d=1$. Hence, $r,r'\in\Delta A$. Now let $x_i$ be the $i$th coordinate of $x$ and suppose $a_i\not\in A(\mu)$. Then $x\in\cD$ implies $x_i\leq 0$, and $d_i=0$ by hypothesis, so $r(a_i)=0$. Likewise, if $a_i\not\in A(\mu')$, then $r'(a_i)=0$ by $x_i\geq 0$ and $d_i=0$. It follows that $(r,r')\in\cR_x$.
\end{proof}

We remark that the decomposition $x=r-r'$ is in fact unique if $x$ has \emph{any} decomposition with $\supp{r}\cap\supp{r'}=\emptyset$, because then $\vsum x^+=1$, or if $A(\mu)\cap A(\mu')$ has at most one element. 

\begin{proof}[{\bfseries Proof of Proposition~\ref{prop:decomp}}]
Consider any $x\in\cD$. By Lemma~\ref{lem:decomp}, there is a decomposition $x=r-r'$ with $(r,r')\in\cR$ if and only if $(r,r')=(x^++d,x^-+d)$ for some $d$ as in Lemma~\ref{lem:decomp}. Suppose this is the case. Analogously, since $x_0=r_0-r_0'\in\cD$, $(r_0,r_0')=(x_0^++d_0,x_0^-+d_0)$ for $d_0=r_0-x_0^+=r_0'-x_0^-$. Thus, $r-r_0=x^++d-x_0^+-d_0$ and $r'-r_0'=x^-+d-x_0^--d_0$, which implies $\norm{r-r_0}\leq\norm{x^+-x_0^+}+\norm{d-d_0}$ and $\norm{r'-r_0'}\leq\norm{x^--x_0^-}+\norm{d-d_0}$. Therefore, it is enough to show that for every (small enough) $\delta>0$ and every $x\in\cD$ with $\norm{x-x_0}<\delta$ there is some $d$ as in Lemma~\ref{lem:decomp} such that $\norm{d-d_0}$ vanishes as $\delta\to 0$.

First suppose $A(\mu)\cap A(\mu')=\emptyset$. For any $x\in\cD$, there is by definition some $(r,r')\in\cR_x$, and then, by Lemma~\ref{lem:decomp}, $(r,r')=(x^++d,x^-+d)$, where now necessarily $d=0$. Analogously, $d_0=0$, so in particular $\norm{d-d_0}=0<\delta$.

Now suppose there is some $a_i\in A(\mu)\cap A(\mu')$. If $d_0\neq 0$, assume w.l.o.g.\ (by Lemma~\ref{lem:decomp}) that $a_i$ is such that the corresponding $i$th coordinate of $d_0$ is positive. Let $d=d_0+\vsum(x_0^+-x^+)e_i$, where $e_i$ denotes the $i$th unit vector. Then $d$ differs from $d_0$ only in the $i$th coordinate. Thus, if $d_0\neq 0$, the choice of $a_i$ implies $d\geq 0$ for all $\delta$ small enough. If $d_0=0$, Lemma~\ref{lem:decomp} implies $\vsum x_0^+=1-\vsum d_0=1$ and, since also $x$ has a corresponding representation with \emph{some} $d\geq 0$, $\vsum x^+\leq 1$, so that also the constructed $d$ satisfies $d\geq d_0=0$. Further, by construction $\vsum d=\vsum d_0+\vsum(x_0^+-x^+)$, so as required $\vsum d=1-\vsum x^+$ by $\vsum d_0=1-\vsum x_0^+$. Now consider any $a_j\not\in A(\mu)\cap A(\mu')$. Then $a_j\neq a_i$, so the $j$th coordinate of $d$ agrees with the $j$th coordinate of $d_0$, and the latter is zero by Lemma~\ref{lem:decomp}. Therefore, $d$ satisfies all properties in Lemma~\ref{lem:decomp} (if $\delta$ is small enough). Finally, $\norm{d-d_0}=\lvert{\vsum(x_0^+-x^+)}\rvert$, which indeed vanishes as $\delta\to 0$.
\end{proof}

\subsection{Proof of Proposition~\ref{prop:srobust}}\label{app:srobust}

The Quasiconcavification Theorem of \citet{lipnowski2020cheap} and their Theorem~1 (Securability) together imply that there exists an equilibrium that achieves the given sender payoff $s$. Because $s\not\in\cV(\mu_0)$, Lemma~LR tells us that any equilibrium achieving $s$ must have $\abs{\supp{p}}\geq 2$. But since $\abs{\Theta}=2$, Lemma~\ref{lem:caratheodory} allows us to assume $\abs{\supp{p}}=2$. Therefore, let $\supp{p}=\{\mu,\mu'\}$ and $\mu<\mu'$. Using once more Lemma~LR, there is an equilibrium with a belief distribution of the present form and payoff $s$ if and only if in fact $\mu<\mu_0<\mu'$ and $s\in\cV(\mu)\cap\cV(\mu')$. This implies that the two sets in \eqref{extremesupport} are nonempty. Further, they respectively have a minimal and a maximal element by upper hemicontinuity of $\cV$ (cf.\ Lemma~\ref{lem:Aneighborhood}), so we may in fact assume $\mu$ and $\mu'$ are as specified in \eqref{extremesupport}.

The present equilibrium is robust if one of the conditions \ref{fullyinf}--\ref{fouractions} in Proposition~\ref{prop:binarybeliefrobust} holds. In particular, condition~\ref{fullyinf} is satisfied if $\mu=0$ and $\mu'=1$, i.e., if $s$ is in both $\cV(0)$ and $\cV(1)$. Therefore, suppose $\mu>0$ or $\mu'<1$. Specifically, assume $\mu'<1$, i.e., $s\not\in\cV(1)$, because the case in which $\mu'=1$ and $\mu>0$ is symmetric to $\mu=0$ and $\mu'<1$.

Using Lemma~\ref{lem:Aneighborhood}, consider any neighborhood $U_\mu$ of $\mu$ small enough such that $A(\tilde{\mu})\subseteq A(\mu)$ for all $\tilde{\mu}\in U_\mu$, where the inclusion must now be strict whenever $\tilde{\mu}<\mu$, because then $s\in\cV(\mu)\setminus\cV(\tilde{\mu})$ by construction of $\mu$. Thus, $\abs{A(\mu)}\geq 2$ if $\mu>0$, and analogously $\abs{A(\mu')}\geq 2$ given $\mu'<1$. If $A(\mu)\cap A(\mu')=\emptyset$, it already follows that $A(\mu)\cup A(\mu')$ has at least three elements and at least four if $\mu>0$, so that condition~\ref{threeactions} or condition~\ref{fouractions} is satisfied.

Now suppose $A(\mu)\cap A(\mu')$ is nonempty, so in fact $A(\mu)\cap A(\mu')=A(\mu_0)$ by Lemma~\ref{lem:Aintersection}. Since $s\in\cV(\mu)$ but $s>\max\cV(\mu_0)$, there must be some $a\in A(\mu)\setminus A(\mu_0)$ such that $v_S(a)\geq s$. Analogously, there is another $a'\in A(\mu')\setminus A(\mu_0)$ such that $v_S(a')\geq s$. It follows that $A(\mu)\cup A(\mu')$ has at least three elements. Thus, condition~\ref{threeactions} holds if $\mu=0$, i.e., if $s\in\cV(0)\setminus\cV(1)$, and by the mentioned symmetry this implies that condition~\ref{threeactions} holds likewise if $s\in\cV(1)\setminus\cV(0)$.

In summary, we have so far established robustness whenever $s$ is in $\cV(0)\cup\cV(1)$, i.e., if condition~\ref{a)} holds, and also whenever $A(\mu)\cap A(\mu')$ is empty. The latter is in fact equivalent to condition~\ref{b)}, because Lemmas \ref{lem:Aneighborhood} and \ref{lem:Aintersection} together imply that $A(\mu)\cap A(\mu')$ is nonempty if and only if $A(\tilde{\mu})=A(\mu_0)$ for all $\tilde{\mu}\in(\mu,\mu')$.

In all remaining cases, in which $s$ is neither in $\cV(0)$ nor in $\cV(1)$, i.e., $\mu>0$ and $\mu'<1$, and $A(\mu)\cap A(\mu')$ is nonempty, the present equilibrium is robust if and only if condition~\ref{fouractions} in Proposition~\ref{prop:binarybeliefrobust} holds, i.e., $\abs{A(\mu)\cup A(\mu')}\geq 4$. For nonempty $A(\mu)\cap A(\mu')$ we have already argued that then $A(\mu)\cap A(\mu')=A(\mu_0)$ and there is some $a\in A(\mu)\setminus A(\mu_0)$ as well as some $a'\in A(\mu')\setminus A(\mu_0)$. It follows that $\abs{A(\mu)\cup A(\mu')}\geq 4$ if and only if $\abs{A(\mu)}\geq 3$ or $\abs{A(\mu')}\geq 3$, which is condition~\ref{c)}.

At this point we have shown that the present equilibrium with $\supp{p}=\{\mu,\mu'\}$ satisfying \eqref{extremesupport} is robust if and only if one of the conditions \ref{a)}, \ref{b)}, or \ref{c)} holds. To prove that this is equivalent to the existence of \emph{any} robust equilibrium achieving the payoff $s$, it is enough to show that the latter implies robustness of the present equilibrium. We only need to consider the case in which $\mu>0$, $\mu'<1$, and $A(\mu)\cap A(\mu')\neq\emptyset$, since otherwise robustness is already established. 

So, assume the initial equilibrium is robust and repeat the beginning of the proof, noting that Lemma~\ref{lem:caratheodory} preserves robustness. Thus, we may assume that also the robust equilibrium has a belief distribution with support $\{\tilde{\mu},\tilde{\mu}'\}$, where $\tilde{\mu}<\tilde{\mu}'$. By construction, $\mu\leq\tilde{\mu}<\tilde{\mu}'\leq\mu'$. This has the following two consequences. First, since $\mu>0$, $\mu'<1$, and the initial equilibrium is robust, condition~\ref{fouractions} in Proposition~\ref{prop:binarybeliefrobust} must hold, i.e., $\abs{A(\tilde{\mu})\cup A(\tilde{\mu}')}\geq 4$. Second, since $A(\mu)\cap A(\mu')\neq\emptyset$, Lemma~\ref{lem:Aintersection} now implies that $A(\tilde{\mu})\subseteq A(\mu)$ and $A(\tilde{\mu}')\subseteq A(\mu')$. Together, it follows that also $\abs{A(\mu)\cup A(\mu')}\geq 4$.

Finally, consider the further sufficient condition, which is $\abs{\set{\tilde{a}\in\bigcup_{\tilde{\mu}\in\Delta\Theta}A(\tilde{\mu})}{v_S(\tilde{a})\leq s}}\geq 2$. To verify it, we again only need to consider $\mu>0$, $\mu'<1$, and $A(\mu)\cap A(\mu')\neq\emptyset$. If $\abs{A(\mu_0)}\geq 2$, then, since there is some $a\in A(\mu)\setminus A(\mu_0)$ and $A(\mu_0)\subseteq A(\mu)$ as argued before, we have $\abs{A(\mu)}\geq 3$, so robustness follows from condition~\ref{c)}. 

Therefore, suppose $A(\mu_0)$ is a singleton and there is another belief $\tilde{\mu}\neq\mu_0$ and an action $\tilde{a}\in A(\tilde{\mu})\setminus A(\mu_0)$ such that $v_S(\tilde{a})\leq s$. By Lemma~\ref{lem:Aintersection} and $A(\mu)\cap A(\mu')\neq\emptyset$, we must have $\tilde{\mu}\leq\mu$ or $\tilde{\mu}\geq\mu'$. By symmetry, assume w.l.o.g.\ the former is the case. Now let $\bar{\mu}$ be the supremum of the (nonempty) set of all such $\tilde{\mu}\leq\mu$. Since $\bar{\mu}\leq\mu$ and, by Lemma~\ref{lem:Aneighborhood}, $A(\tilde{\mu})\subseteq A(\bar{\mu})$ for all $\tilde{\mu}$ in some neighborhood of $\bar{\mu}$, $\bar{\mu}$ belongs itself to the set, i.e., it is its maximal element. Moreover, we must in fact have $\bar{\mu}=\mu$. To show this, suppose by way of contradiction $\bar{\mu}<\mu$. Then, by construction of $\mu$, $\max\cV(\bar{\mu})<s$ and hence also $\max\cV(\tilde{\mu})<s$ for all $\tilde{\mu}$ in some neighborhood of $\bar{\mu}$ by $A(\tilde{\mu})\subseteq A(\bar{\mu})$. For any $\tilde\mu\in(\bar{\mu},\mu]$, however, $\max\cV(\tilde{\mu})<s$ only if $A(\tilde{\mu})\subseteq A(\mu_0)$, but the latter is not possible for any $\tilde{\mu}<\mu$. Indeed, if $\tilde{\mu}<\mu$, then necessarily $A(\tilde{\mu})\cap A(\mu_0)=\emptyset$, because otherwise Lemma~\ref{lem:Aintersection}, $\tilde{\mu}<\mu<\mu_0<\mu'$, and $A(\mu_0)\subseteq A(\mu')$ would imply $A(\mu)=A(\mu_0)$ in contradiction to the previously established $a\in A(\mu)\setminus A(\mu_0)$.

Having shown that $\bar{\mu}=\mu$, we may conclude there is in particular some $\tilde{a}\in A(\mu)\setminus A(\mu_0)$ such that $v_S(\tilde{a})\leq s$. In addition, we may assume this $\tilde{a}$ to be distinct from the given $a\in A(\mu)\setminus A(\mu_0)$ with $v_S(a)\geq s$. Indeed, consider any $\tilde{\mu}<\mu$ but in a neighborhood of $\mu$ such that $A(\tilde{\mu})\subseteq A(\mu)$. Since $\tilde{\mu}<\mu$, $A(\tilde{\mu})\cap A(\mu_0)=\emptyset$ as argued before and $s\not\in\cV(\tilde{\mu})$ by construction of $\mu$. It follows that we may take either $a$ or $\tilde{a}$ to be from $A(\tilde{\mu})$ and the corresponding inequality to be strict, which yields $v_S(\tilde{a})<s\leq v_S(a)$ or $v_S(\tilde{a})\leq s<v_S(a)$. The fact that $\tilde{a}$ is distinct from $a$ finally implies that $\abs{A(\mu)}\geq 3$, and now robustness follows from condition~\ref{c)}.
\qed

\subsection{Proof of Proposition~\ref{prop:fullyrobust}}\label{app:fullyrobust}

We prove the contrapositive. Let $\Theta=\{\theta_1,\theta_2\}$, $u_R$ be an arbitrary receiver utility function, $u_S$ a state-independent sender utility function with $u_S(\theta_1)=u_S(\theta_2)=v_S\in\R^n$, and $(p,\rho)$ an equilibrium such that the (constant) induced sender payoff $s^*$ is not in $\cV(\mu_0)$. Using Lemma~\ref{lem:caratheodory}, we can assume that the support of $p$ contains exactly two beliefs $\mu$ and $\mu'$ such that $\theta_1\in\supp{\mu}$ and $\theta_2\in\supp{\mu'}$. We will show that there exists an arbitrarily nearby utility function for the sender such that every corresponding equilibrium with belief distribution $p$ induces a sender payoff in $\cV(\mu_0)$. By Lemma~\ref{clm:samebeliefs}, this suffices.  

For $p$ to be an equilibrium belief distribution for a perturbed sender utility function $\tilde u_S$, there must be some
\begin{equation*}
x\in\cD=\setd{r-r'}{\text{$r,r'\in\Delta A$, $\supp{r}\subseteq A(\mu)$, and $\supp{r'}\subseteq A(\mu')$}}
\end{equation*}
such that $\tilde u_S(\theta_1)\cdot x\geq 0\geq\tilde u_S(\theta_2)\cdot x.$ 
Let $d$ be the vector in $\R^n$ whose $i$th entry is $1$ if $a_i\in A(\mu)\setminus A(\mu')$, $-1$ if $a_i\in A(\mu')\setminus A(\mu)$, and $0$ if $a_i\in A(\mu)\cap A(\mu')$. Then, for any $x=r-r'\in\cD$,
\begin{equation*}
d\cdot x=\sum_{a\in A(\mu)\setminus A(\mu')}r(a)~~+\sum_{a\in A(\mu')\setminus A(\mu)}r'(a)\geq 0.
\end{equation*}
We can approximate $u_S$ by a sender utility function $\tilde u_S$ such that $\tilde u_S(\theta_1)=v_S$ and $\tilde u_S(\theta_2)=v_S+\epsilon d$ for some $\epsilon>0$. Take $x=r-r'\in\cD$ corresponding to some equilibrium for the perturbed sender utility function $\tilde u_S$ with belief distribution $p$. Then the condition $\tilde u_S(\theta_1)\cdot x\geq 0\geq\tilde u_S(\theta_2)\cdot x$ implies $d\cdot x\leq 0$, which again implies
\begin{equation*}
\sum_{a\in A(\mu)\setminus A(\mu')}r(a)~~+\sum_{a\in A(\mu')\setminus A(\mu)}r'(a)=0.
\end{equation*}
So both $r$ and $r'$ are mixtures over actions in $A(\mu)\cap A(\mu')$. This immediately implies $d\cdot r'=0$ and additionally, by Lemma~\ref{lem:Aintersection}, that both $r$ and $r'$ are supported on $A(\mu_0)$. Consequently, both $v_S\cdot r$ and $v_S\cdot r'$ are in $\cV(\mu_0)$. That the same holds for the sender's equilibrium payoff in each state then follows from the construction of $\tilde u_S$, the fact that we can assume $r$ ($r'$) to be the receiver's equilibrium reply to the posterior belief $\mu$ ($\mu'$), and $d\cdot r'=0$.
\qed
\bibliography{references.bib}

\end{document}